\RequirePackage{silence}
\documentclass{article}

\usepackage[main, final]{neurips_2026}
\usepackage[utf8]{inputenc} 
\usepackage[T1]{fontenc}    
\usepackage{hyperref}       
\usepackage{url}            
\usepackage{booktabs}       
\usepackage{amsfonts}       
\usepackage{nicefrac}       
\usepackage{microtype}      
\usepackage{xcolor}        

\usepackage{amsmath}
\usepackage{amsthm}
\usepackage{amssymb}
\usepackage{tikz}
\usepackage{subcaption}
\newtheorem{theorem}{Theorem}
\newtheorem{lemma}[theorem]{Lemma}
\newtheorem{definition}[theorem]{Definition}
\newtheorem{corollary}[theorem]{Corollary}
\newtheorem{proposition}[theorem]{Proposition}

\newcommand{\metric}[0]{NMD}
\newcommand{\longmetric}[0]{Neural Manifold Diversity}

\title{
Events as Triggers for Behavioral Diversity in Multi-Agent Reinforcement Learning
}

\author{
  Hannes B\"uchi \hspace{0.5cm} \textbf{Manon Flageat}$^*$ \hspace{0.5cm} \textbf{Eduardo Sebasti\'an}$^*$ \hspace{0.5cm} \textbf{Amanda Prorok}\\
  University of Cambridge\\
  \texttt{\{hmb71, mf873, es2121, asp45\}@cl.cam.ac.uk} | $^*$Equal co-supervision
}

\begin{document}

\maketitle

\begin{abstract}
Effective multi-agent cooperation requires agents to adopt diverse behaviors as task conditions evolve—and to do so at the right moment. Yet, current Multi-Agent Reinforcement Learning (MARL) frameworks that facilitate this diversity are still limited by the fact that they bind fixed behaviors to fixed agent identities. Consequently, they are ill-equipped for tasks where agents need to take on different roles at very specific moments in time. We argue that, to define these behavioral transitions, the missing ingredient is \textbf{events}. Events are changes in the state of the system that induce qualitative changes in the task. Based on this view, we introduce a framework that decouples agent identity from behavior, capturing a continuous manifold from which agents instantiate their behaviors in response to events. This framework is based on two elements. First, to build an expressive behavior manifold, we introduce Neural Manifold Diversity (NMD), a formal distance metric that remains well-defined when behaviors are transient and agent-agnostic. Second, we use an event-based hypernetwork that generates Low-Rank Adaptation (LoRA) modules over a shared team policy, enabling on-the-fly agent-policy reconfiguration in response to events. We prove that this construction ensures that diversity does not interfere with reward maximization by design. Empirical results demonstrate that our framework outperforms established baselines across benchmarks while exhibiting zero-shot generalization, and being the only evaluated method that solves tasks requiring sequential behavior reassignment.

\end{abstract}

\section{Introduction}
\label{sec:introduction}

\begin{figure}[htbp]
    \centering \includegraphics[width=\linewidth]{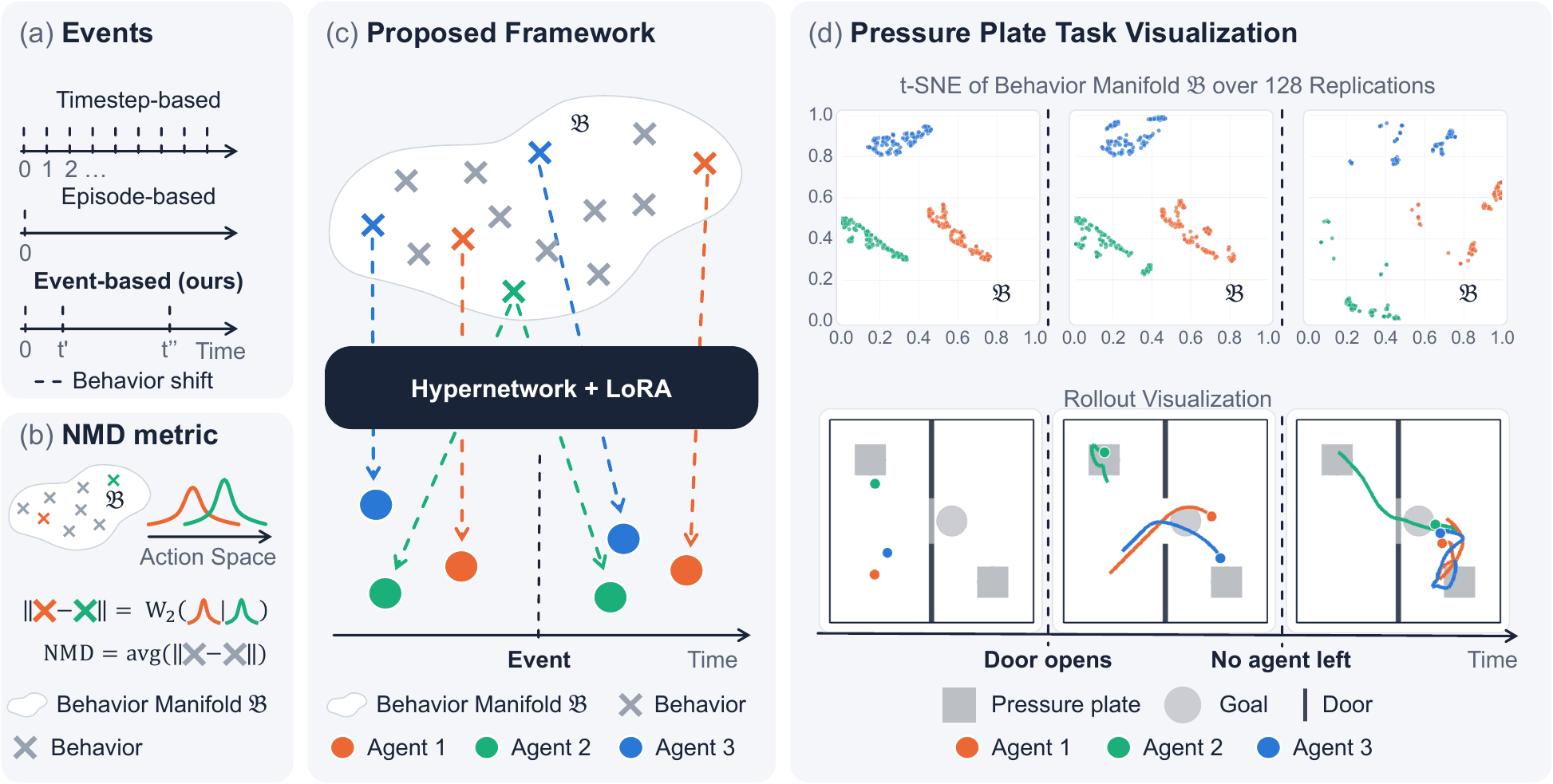}
    \caption{
    \textbf{Proposed Framework.} \textbf{(a)} Rather than shifting behaviors on a fixed timestep or episode schedule, our framework triggers behavioral transitions in response to task events. We realize this through two components: \textbf{(b)} \metric{}, a diversity metric defined directly on the behavior manifold and independent of which agent executes which behavior; and \textbf{(c)} an event-driven hypernetwork that generates LoRA~\citep{hu2021lora} modules over a shared policy to allocate behaviors from the manifold. \textbf{(d)} We visualize the behavior manifold via a two-dimensional t-SNE projection across $128$ replications of the Pressure Plate environment, alongside agent trajectories for one rollout.
    }
    \label{fig:main}
\end{figure}

Effective cooperation in multi-agent systems hinges less on agent identity than on agent behavior at any given moment. As tasks unfold, optimal team performance demands that agents adapt their behaviors dynamically, not uniformly across the team nor on a synchronized schedule, but at the right moment and under the right circumstances. 
When a quadrotor in a search team loses power, the remaining agents expand their coverage to compensate; or a defender in football pushes forward only once possession changes to the other team. 
These transitions are not arbitrary but triggered by events. 
Events are changes in the system's state---e.g., team composition and environment configuration---that induce qualitative changes in the task and demand a behavioral response. 

Behavioral diversity, or behavioral heterogeneity, is a key driver of performance in cooperative Multi-Agent Reinforcement Learning (MARL)~\citep{li2021celebrating}, and recent approaches have proposed various ways to introduce it within agent teams~\citep{bettini2023heterogeneous}.
Yet few existing approaches capture dynamic behavior adaptation. Most either assign behaviors statically, fixing each agent's role for the entire episode, or periodically, switching every few time steps regardless of context. 
We argue that this gap reflects a deeper conceptual issue. Whether diversity is enforced via unique per-agent heterogeneous components in the policy~\citep{bettini2024controlling}, agent-specific weights generated by a hypernetwork~\citep{Tessera2024}, or partitioned action spaces~\citep{wang2020rode}, existing approaches share the same implicit assumption: behaviors are tied to specific agents.

We challenge this assumption. We view behaviors as a property of the task---a continuous manifold the team draws from---whose traversal is dictated by events rather than by agent identity. Decoupling behavior and identity reframes the central question of diversity in MARL: rather than asking how to differentiate agents, we ask how to populate, structure and navigate the behavior manifold.

This brings three requirements. First, to ensure the manifold is genuinely diverse, we need a metric defined directly on behaviors rather than on agents. We introduce \longmetric{} (\metric{}), which measures diversity in the behavior manifold directly, independent of the agents currently executing those behaviors. Second, we need an architecture that can both learn such a manifold and adapt agents as events unfold. We achieve this through an event-based hypernetwork that generates Low-Rank Adaptation (LoRA)~\citep{hu2021lora} modules over a shared team policy, a computationally efficient alternative to full policy regeneration. 
Third, we need an optimization procedure that jointly maximizes team reward while preserving manifold diversity. We adopt a diversity-constrained objective~\citep{bettini2024controlling}, and prove that this construction does not interfere with reward maximization by design. 
Fig.~\ref{fig:main} illustrates our framework.
In summary, our contributions are:
\begin{itemize}
    \item We introduce an event-centric paradigm in which agents draw their behaviors dynamically from a continuous manifold, adapting in response to events.
    \item We propose a behavior-centric metric, \metric{}, that models diversity directly within this manifold, independent of which agent executes which behavior.
    \item We realize the framework through an event-driven hypernetwork generating LoRA modules. We prove that this architecture aligns task-objective gradients with \metric{} constraints through diversity control, guaranteeing joint reward maximization and diversity preservation.
\end{itemize}

\section{Related Work}\label{sec:related}

A variety of MARL approaches investigate the allocation of diverse behaviors (also frequently referred to as ``roles'') within a team. We review the most closely related work along three axes, and refer the reader to Appendix~\ref{app:relatedwork} for an extended discussion.

\textbf{Temporal Behavioral Assignment.}
A core question for behavior allocation is the assignment timeline. Existing approaches fall into two categories: methods that fix one behavior per agent for the entire episode~\citep{Tessera2024, bettini2024controlling, jiang2021emergence, li2021celebrating}, and methods that adapt behaviors on a fixed temporal schedule regardless of task conditions~\citep{Fu2025, qi2026rois, mahajan2019maven, wang2020roma, goel2025r3dm}. We advocate for a third category: event-based allocation, in which behavioral shifts are triggered by state changes rather than fixed schedules or episode boundaries. Event-triggering itself has a long history in control~\citep{tabuada2007event, heemels2012introduction}, reinforcement learning~\citep{baumann2018deep}, MARL communication~\citep{hu2023etcnet}, and multi-robot task allocation~\citep{alliance1998, gerkey2004formal}. There, the event decides {when} a given controller is updated or a message is sent. In our framework, the event decides {which} behavior is instantiated from a continuous, learned behavior manifold whose diversity is measured and constrained.

\textbf{Representations of Behaviors.}
Behavioral specialization is fundamentally constrained by the underlying behavior representation. 
Existing approaches include learning explicit per-behavior policies or policy deviations~\citep{jiang2021emergence, li2021celebrating, bettini2024controlling, christianos2021seps, kim2023networkpruning, li2024adaptive}, conditioning a shared policy on agent identity or capabilities~\citep{Tessera2024, Fu2025, deka2021natural}, and learning a latent behavior space to condition a shared policy or hypernetwork~\citep{mahajan2019maven, wang2020roma}.
Our framework builds on similar representations but conditions behavior generation on task events, enabling dynamic allocation that prior approaches cannot support.

\textbf{Quantifying and Enforcing Behavioral Diversity.}
Existing methods to encourage diversity fall into two categories: auxiliary rewards or losses~\citep{wang2020roma, li2021celebrating, mahajan2019maven, jiang2021emergence}, which provide no formal guarantee on final team diversity; and hard constraints~\citep{bettini2024controlling, tan2023policy}, which provide stronger guarantees but require a target diversity value a priori. Diversity is commonly quantified via information-theoretic objectives~\citep{wang2020roma, li2021celebrating, mahajan2019maven, jiang2021emergence} or explicit metrics such as System Neural Diversity (SND)~\citep{bettini2023system}.
We propose \metric{}, a new metric inspired by SND that remains well-defined under dynamic behavior allocation, and use it within a constrained optimization framework to enforce target diversity levels.

\section{Background}
\label{sec:background}

\textbf{Cooperative Multi-Agent Reinforcement Learning.} We model cooperative multi-agent tasks as Partially Observable Markov Games (POMGs)~\citep{shapley1953markovgames}, defined by the tuple \mbox{$\mathcal{M} = \langle \mathcal{N}, \mathcal{S}, \mathcal{A}, \mathcal{O}, \mathcal{P}, \Omega, \mathcal{R}, \gamma, \mathcal{B}\rangle$}. Here, $\mathcal{N}$ is the set of agents, $\mathcal{S}$ the state space, $\mathcal{A}$ and $\mathcal{O}$ the joint action and observation space, $\mathcal{P}$ the transition dynamics, $\Omega$ the observation function, $\mathcal{R}$ the shared reward function, and $\gamma \in [0,1)$ is the discount factor. We additionally let $\mathcal{B}$ denote the \textbf{behavior manifold} of $\mathcal{M}$, i.e., the space of admissible policies $\pi: \mathcal{O} \rightarrow \mathcal{P}(\mathcal{A})$ that the task supports. In this sense, $\mathcal{B}$ is a property of the task---its structure and coverage are determined by what strategies $\mathcal{M}$ admits---and is not given a priori; how to represent and learn it is the central concern of Section \ref{sec:framework}. Each agent $i \in \mathcal{N}$ executes a policy $\pi_i(a_{i,t} | o_{i,t}) \in \mathcal{B}$, and the team objective is to find a joint policy $\boldsymbol{\pi}^*$ maximizing the expected discounted return $R = \mathbb{E}[\sum_{t=0}^{\infty}\gamma^t\sum_{i\in\mathcal{N}}r_{i,t}]$. Note that a consequence of this formulation is that a behavior is equivalent to a policy, so we use these terms as interchangeable. We also remark that the observation vector $o_i \in \mathcal{O}$ includes a physical capability descriptor when relevant for the task, in contrast to existing approaches where physical capabilities are defined as independent quantities. We do this to emphasize that our focus is on behavioral diversity, in which agents display different action distributions given identical inputs, rather than physical diversity, which is thus a special case of our formulation. 

\textbf{Events as Triggers of Behavioral Change.} Standard POMGs treat all transitions as instances of the same dynamics $\mathcal{P}$, leaving no formal distinction between qualitative changes in task progress. To capture such moments, we augment the POMG with a discrete \textbf{event space} $\Xi$ and an \textbf{event-detection function} $\mathcal{E}: \mathcal{S} \rightarrow \Xi \cup \{\varnothing\}$ that maps the current state of the system to either a triggered event $\xi_t \in \Xi$ or the null symbol $\varnothing$, i.e., no event happened. 
Roughly speaking, an event is a change in the system state---including team composition and environment configuration---that induces a qualitative shift in the optimal joint behavior.
Events may originate from distinct sources. In our empirical study, we consider the following cases: changes in an agent's own observations (e.g., a LiDAR detecting a target); changes in the team configuration (e.g., an agent being added); or environmental signals (e.g., a door opening). When an event triggers, it induces a non-stationary transition $(s_t, \mathcal{N}_t, \mathbf{b}_t) \xrightarrow{\xi_t} (s_{t+1}, \mathcal{N}_{t+1}, \mathbf{b}_{t+1})$ that displaces the system onto a new optimal trajectory manifold, where $\mathbf{b}_t = (\pi_{1,t}, \ldots, \pi_{|\mathcal{N}_t|,t}) \in \mathcal{B}^{|\mathcal{N}_t|}$. The challenge for the team is to respond to $\xi_t$ by adapting the behaviors of its members. We refer to the resulting setting as an \textbf{Event-Augmented POMG}, defined, with an abuse of notation, as $\mathcal{M} = \langle \mathcal{N}, \mathcal{S}, \mathcal{A}, \mathcal{O}, \mathcal{P}, \Omega, \mathcal{R}, \gamma, \mathcal{B}, \Xi, \mathcal{E}\rangle$.

\textbf{Problem Formulation.} Given $\mathcal{M}$, our goal is to learn a mechanism that decouples agent identity from executed behavior. Concretely, we seek to learn (1) the continuous behavior manifold $\mathcal{B}$ that captures the space of behaviors admitted in a task, and (2) an event-conditioned mechanism $f: \mathcal{O} \times \Xi \rightarrow \mathcal{B}$ that maps each agent's observation and most recent event to a behavior $\pi_{i,t} \in \mathcal{B}$.
Beyond maximizing cumulative reward, if behaviors are truly properties of the task only, the mechanism should generalize zero-shot as any of the following vary: (1) number of agents, (2) physical capabilities, and (3) event sequences, including disruptions unseen during training (e.g., abrupt teammate removal).

\section{A Framework for Event-Driven Behavioral Adaptation}
\label{sec:framework}

To describe our framework, we proceed in three steps. First, we introduce \metric{} (\longmetric{}), the metric that characterizes the behavior manifold $\mathcal{B}$ (Section \ref{subsec:nmd}). Second, we present its architectural realization: an event-driven hypernetwork that generates LoRA modules over a shared policy backbone that allows us to learn $\mathcal{B}$ and allocate policies to agents according to it (Section \ref{subsec:architecture}). Third, we describe how diversity is maintained and controlled on $\mathcal{B}$ and we state the theoretical guarantees that follow from this construction, deferring proofs to Appendix \ref{app:proofs} (Section \ref{subsec:dico}).  

\subsection{Modeling the Behavior Manifold: Neural Manifold Diversity}\label{subsec:nmd}
The behavior manifold $\mathcal{B}$ is the central object our framework manipulates: agents draw from it, events trigger movement within it, and its geometry determines what coordination strategies the team has available. To make this manifold a tractable object---one we can measure and control---we begin by quantifying the diversity of policies drawn from it.
\begin{definition}
Let $\mathcal{B}$ denote the continuous behavior manifold induced by the task $\mathcal{M}$ and let $p(\pi)$ denote a probability density over $\mathcal{B}$. The \longmetric{} is the expected pairwise divergence between behaviors sampled from $\mathcal{B}$, evaluated over the observation distribution $p(o)$:
\begin{equation}\label{eq:og_nmd}
    \mathrm{\metric{}}(\mathcal{B}, \mathcal{O}) = \int_{\mathcal{B}} \int_{\mathcal{B}} \int_{\mathcal{O}} W_2\!\left(\pi_m(o), \pi_n(o)\right) p(o)\, do\, p(\pi_m)\, d\pi_m\, p(\pi_n)\, d\pi_n,
\end{equation}
where $W_2$ is the 2-Wasserstein distance between action distributions.
\end{definition}
A property that justifies treating \metric{} as a principled diversity measure rather than an arbitrary statistic is that the pairwise term in Eq. \eqref{eq:og_nmd} defines a proper distance on $\mathcal{B}$.
\begin{proposition}\label{prop:nmd_distance}
    Let $d: \mathcal{B} \times \mathcal{B} \to \mathbb{R}_{\geq 0}$ be $d(\pi_m, \pi_n) := \mathbb{E}_{o \sim p(o)}\!\left[W_2(\pi_m(o), \pi_n(o))\right]$. Then $d$ is a pseudometric on $\mathcal{B}$, and a metric whenever policies are distinguishable on the support of $p(o)$.
\end{proposition}
The proposition guarantees that $d$ satisfies non-negativity, symmetry, identity of indiscernibles (up to $p(o)$ not belonging to a null set), and the triangle inequality. As a consequence, $\mathrm{\metric{}}(\mathcal{B}, \mathcal{O})$ admits a clear interpretation: it is the expected $d$-distance between two behaviors drawn independently from $\mathcal{B}$. The problem with Eq. \eqref{eq:og_nmd} is that it is intractable because $\mathcal{B}$ and $\mathcal{O}$ are continuous and high-dimensional. In practice, we instantiate a finite collection of policies $\{\pi_m\}_{m=1}^B$ of behaviors and we observe a finite set of episodic observations $\mathcal{O}_{\mathrm{ep}} \subset \mathcal{O}$, substituting the integrals with empirical averages to yield to an estimator:
\begin{equation}\label{eq:nmd_approx}
    \widehat{\mathrm{\metric{}}}\!\left(\{\pi_m\}_{m=1}^B\right) = \frac{2}{B(B-1)\,|\mathcal{O}_{\mathrm{ep}}|} \sum_{m=1}^{B}\sum_{n=m+1}^{B} \sum_{o \in \mathcal{O}_{\mathrm{ep}}} W_2\!\left(\pi_m(o), \pi_n(o)\right).
\end{equation}
This is an unbiased statistic for $\mathrm{\metric{}}(\mathcal{B}, \mathcal{O})$ whenever the sampled behaviors are i.i.d. In our implementation, $\{\pi_m\}_{m=1}^B$ is the set of behaviors in a parallel batch of environments and $\mathcal{O}_{\mathrm{ep}}$ is the set of observations encountered along the rollout.

Structurally, Eqs. \eqref{eq:og_nmd} and \eqref{eq:nmd_approx} resemble SND~\citep{bettini2023system}, as both use the pairwise $W_2$ distance between action distributions. The novelty of \metric{} lies not in this distance kernel but in the object it is applied to. SND is a function of an agent-indexed tuple $(\pi_1, \ldots, \pi_N) \in \mathcal{B}^N$, which means it computes pairwise distances between policies indexed by agent
identities. In contrast, \metric{} is a functional of a probability measure $p$ over the behavioral manifold $\mathcal{B}$, which implies that \metric{} computes a distance that characterizes the behavioral manifold itself at any given time, irrespective of the agents' identities, number or order.
Appendix~\ref{app:nmd_snd} formalizes this relationship, including an example in which SND reports homogeneity of policies even though the policies are diverse for the entire rollout.

\subsection{Architectural Realization}\label{subsec:architecture}

To capture, populate and navigate $\mathcal{B}$, we require an architecture satisfying three requirements: \textit{(i)} it must produce behaviors as a function of an agent's local observation and the most recent event; \textit{(ii)} it must do so cheaply enough to avoid prohibitive per-timestep computation; and \textit{(iii)} it must allow to introduce mechanisms to preserve the diversity of the behavior manifold $\mathcal{B}$. We achieve this with an event-driven hypernetwork that emits LoRA modules over a shared policy backbone.

\textbf{Event-driven hypernetwork.} Let $g_\theta$ be a transformer-based hypernetwork. Given the most recent event $\xi_t$, $g_\theta$ ingests a token per agent from their local observations $o_{i,t}$, the event encoding $e_t$, and the global target diversity scalar $\mathrm{\metric{}}_{\mathrm{des}}$, and emits a LoRA pair $(C_m, D_m)$ per agent. Note that we use index $m$ instead of $i$ to highlight that the hypernetwork might predict the same $(C_m, D_m)$ pair for various agents. The hypernetwork is queried only when an event fires (and once at initialization); between events, the pair is held fixed. This avoids the static brittleness of episode-initial generation~\citep{Tessera2024} and the cost and instability of per-step regeneration~\citep{Fu2025}.

\textbf{LoRA over a shared backbone.}  The team shares a policy backbone $W_{\mathrm{shared}}$ and applies the assigned LoRA adapters $D_m C_m$ to the final linear layer of the policy. Given the penultimate feature $\phi(o_t)$, the pre-activation output of behavior $m$ is:
\begin{equation}\label{eq:lora_adaptation}
    z_m(o_t) = W_{\mathrm{shared}}\, \phi(o_t) \;+\; \alpha\, D_m C_m\, \phi(o_t),
\end{equation}
where $\alpha$ is a scalar that is used to preserve diversity, which will be explained in Section \ref{subsec:dico}. The final action distribution $\pi_m(o_t)$ is obtained by applying an activation function (squashed Gaussian for continuous control) to $z_m(o_t)$. Three design choices deserve note. First, restricting the hypernetwork to emit $C_m \in \mathbb{R}^{r \times d}$, $D_m \in \mathbb{R}^{d_a \times r}$ with $r \ll d$ stabilizes optimization relative to generating dense matrices; it also reduces the hypernetwork's output dimensionality (see Appendix \ref{app:lorarank}). Second, applying the LoRA update to the final linear layer ensures that the heterogeneous component enters the policy linearly in $\phi(o_t)$, the condition that makes diversity maintenance possible and gradient projections analytically tractable (Section \ref{subsec:dico}). Third, $W_{\mathrm{shared}}$ accumulates task-general representations, while $(C_m, D_m)$ encodes only the deviation needed to instantiate a specific behavior, keeping the hypernetwork output small and the solution memory efficient. 
Overall, the event-based hypernetwork with LoRA modules allows us to learn a representation of $\mathcal{B}$ through the weights $\theta$ of $g_{\theta}$ and, at the same time, assign behaviors to agents by generating $(C_m, D_m)$ pairs, effectively allowing to capture, populate and navigate the behavior manifold $\mathcal{B}$.

\subsection{Diversity Control on the Behavior Manifold}\label{subsec:dico}

To ensure that the hypernetwork learns a well-structured behavior manifold, we opt to enforce a diversity constraint inspired by Diversity Control (DiCo)~\citep{bettini2024controlling}. 
We express each behavior as the sum of a shared component $W_{\mathrm{shared}} \phi(o_t)$ and a per-behavior deviation \mbox{$u_m(o_t) = D_m C_m \phi(o_t)$}, as
described in Eq. \eqref{eq:lora_adaptation}. 
To ensure diversity constraint, the scalar $\alpha$ is set as:
\begin{equation}\label{eq:alpha}
    \alpha = \frac{\mathrm{\metric{}}_{\mathrm{des}}}{\widehat{\mathrm{\metric{}}}\!\left(\{u_m(o_t)\}_{m=1}^B\right)},
\end{equation}
where the denominator is computed over the $B$ behaviors present in the parallel batch of environments, and $\mathrm{\metric{}}_{\mathrm{des}}$ is the target diversity level. 
Because the LoRA update enters Eq. \eqref{eq:lora_adaptation} linearly, scaling the deviations by $\alpha$ scales the resulting behavioral distance by $\alpha$ exactly, so the realized diversity equals $\mathrm{\metric{}}_{\mathrm{des}}$ by construction (the denominator is smoothed across batches using a moving average for variance reduction; Appendix \ref{app:expdetails} gives the details). 
This linearity is the technical reason for the final-layer linear placement in Section \ref{subsec:architecture}: any nonlinearity between the LoRA update and the policy mean would forfeit this reduction. 
Our formulation is an extension of DiCo~\citep{bettini2024controlling}, which expresses each agent's policy as $\pi_i(o) = \pi_h(o) + \lambda\, \pi_{h,i}(o)$, where $\pi_h$ is a shared component, $\pi_{h,i}$ is a per-agent deviation, and $\lambda$ scales the deviation to match a target SND value.
However, unlike ours, DiCo is agent-centric: each $\pi_{h,i}$ is permanently bound to an agent and the matched diversity is measured between fixed agent indices.
Together, Eqs. \eqref{eq:lora_adaptation} and \eqref{eq:alpha} generalize DiCo from a fixed set of indexed agents to the time-varying set of behaviors instantiated on the manifold.

This brings to the main result that justifies the architectural choice in Section \ref{subsec:architecture} and the diversity control mechanism: we do not merely enforce a target diversity level, but also make reward optimization and diversity preservation jointly stable without auxiliary losses or trust regions. 
\begin{theorem}\label{theorem:joint}
    The gradient of the expected return $R$ with respect to the deviation $u_m(o_t)$ satisfies $\nabla_{u_m} R = \alpha\, P_{u_m}\, (\nabla_{z_m} R)^\top$, where $P_{u_m} = I - \frac{u_m \otimes \nabla_{u_m} \widehat{\mathrm{\metric{}}}}{\widehat{\mathrm{\metric{}}}}$ is an idempotent projection matrix.
\end{theorem}
The theorem shows that diversity control does not enter as an external regularizer, but as a projector on the reward-maximizing gradient: any component of $\nabla_{z_m}R$ that would alter $\widehat{\mathrm{\metric{}}}$ is filtered out before the LoRA parameters are updated. The reward objective and the diversity constraint are therefore aligned by construction---the optimized architecture can move freely along directions that preserve realized diversity, and is gated along directions that would violate it. Theorem \ref{theorem:joint} holds under the assumptions stated in Appendix \ref{app:proofs}, namely, that policies have a shared covariance (Gaussian case), the deviation enters the pre-activation linearly (which the final-layer placement guarantees) and any further nonlinearity operates in its linear regime (typically, the squashing nonlinearity). Analysis of scenarios outside these assumptions can be found in Appendix \ref{app:proofs}.
\begin{corollary}\label{corollary:limits}
    The projection $P_{u_m}$ has two limiting regimes. \textbf{Vanishing deviation}: when $\|u_m\| \to 0$, $P_{u_m} \to I$ and the gradient is unprojected, allowing the behavior to escape a degenerate niche. \textbf{Dominant deviation}: when $\|u_m\|$ is large, $P_{u_m}$ contracts strongly along $u_m$, suppressing updates that would further amplify $\|u_m\|$ and collapse the heterogeneity of the remaining behaviors.
\end{corollary}

\section{Experiments}
\label{sec:experiments}

We empirically address four questions. 
First, we assess whether our framework constitutes a competitive MARL approach: \textbf{(Q1) How does our framework perform relative to established baselines across benchmarks?} 
Second, if behaviors are truly properties of the task, our approach should generalize without retraining across agent counts and capabilities. Our constrained formulation further enables generalization to a third axis, target diversity levels: \textbf{(Q2) Does the framework generalize zero-shot across varying agent counts, capabilities, and diversity levels?}
Third, we verify that our event-based framework triggers the intended behavioral shifts: \textbf{(Q3) Can our framework solve tasks that require sequential behavior allocation in response to events?} 
Fourth, if behaviors are task properties, the framework should also generalize to unseen event sequences, such as the removal of an agent mid-episode: \textbf{(Q4) Does our framework generalize to unseen event sequences?}

\subsection{Experimental Setup}

\textbf{Baselines.} We evaluate our framework against four baselines. \textbf{HyperMARL}~\citep{Tessera2024} allocates behaviors episode-wise via a hypernetwork conditioned on agent identities. 
Capability-Aware Shared Hypernetworks (\textbf{CASH})~\citep{Fu2025} allocates behaviors at every timestep via a hypernetwork conditioned on capabilities when available, or agent identities. 
Diversity Control  (\textbf{DiCo})~\citep{bettini2024controlling} inspired our constrained diversity optimization, but operates on fixed agent indices using cross-agent SND as the constraint; for each task, we perform a grid search over target SND values and report the best-performing variant. 
Finally, Parameter Sharing (\textbf{PS})~\citep{gupta2017cooperative} is a baseline with no behavior allocation, where all agents share the same policy. HyperMARL and CASH represent the extreme cases of triggering schedule of our framework: they allocate behaviors once at $t=0$ and at every step, respectively, while our framework re-allocates them just when events occur. Other event-triggered control and classical task-allocation methods are not directly applicable as baselines, since they require a prescribed controller per mode or a hand-specified task set with utilities (Appendix~\ref{app:relatedwork}).

\textbf{Environments.} 
We consider two sets of experiments. 
First, to address \textbf{(Q1)} and \textbf{(Q2)}, we evaluate our approach on four benchmarks established in recent studies on behavioral diversity~\citep{Tessera2024, bettini2024controlling, bettini2025impact}: (i) \textbf{Navigation} (local-sensing exploration), (ii) \textbf{Dispersion} (reaching $M$ targets from a central point), (iii) \textbf{Reverse Transport} (collective package pushing), and (iv) \textbf{Football} (competitive multi-team benchmark inspired by Google Research Football~\citep{kurach2020google}). 
Second, to address \textbf{(Q3)} and \textbf{(Q4)}, we introduce two custom tasks designed to highlight the benefits of our event-driven formulation. \textbf{Pressure Plate} requires a team to coordinate sequentially by holding down plates to open a locked door separating them from their goal; this task can only be solved with sequential event-based behavior allocation. \textbf{Wind Flocking} requires the team to maintain a formation where an agent must provide shielding for another agent against wind resistance; this task allows for direct visualization of the effects of our proposed \metric{}.
All environments are implemented in VMAS~\citep{bettini2022vmas}, a popular and scalable MARL simulator. Details are in Appendix~\ref{app:expdetails} and source code at \url{https://h-buechi.github.io/event-triggered-diversity/}. 

\textbf{Events.} A key dimension of our framework is the event. In our experiments, an event is defined by the removal of an agent, a modification in agent capabilities, a change in the target diversity, or a change in an element of the environment (e.g., door open). In tasks featuring partial observability, such as Navigation, a change in the local LiDAR reading is likewise categorized as an event. 

\subsection{Q1: Performance compared to established baselines}\label{subsec:q1}

To address \textbf{(Q1)}, we evaluate the performance of all approaches across the four primary tasks. Fig.~\ref{fig:four_tasks_row} reports the task-specific completion rate, the average reward, and the episode length. Our proposed framework systematically ranks top-1, achieving a completion rate of at least $95\%$ across all tasks and significantly outperforming all baselines on the more complex Football task.

Examining each panel reveals how the performance gap scales with task complexity. On Reverse Transport, all methods reach competitive completion rates except \textbf{PS}. The clear separation in this task is in episode length, where our method consistently terminates earliest, closely followed by \textbf{DiCo} but with a wide gap with respect to the other methods. This indicates that event-driven behavior allocation produces more direct trajectories rather than merely succeeding more often. On Navigation, only our framework and \textbf{CASH} reach high completion with low episode length. On Dispersion, the gap widens. Our framework retains a near-perfect completion rate while \textbf{CASH} drops to roughly 50\%, \textbf{DiCo} to roughly 80\%, and \textbf{HyperMARL} and \textbf{PS} collapse. This is the first task where static behavior assignment becomes a clear liability: with M targets and no per-target specialization unlocked at the right moment, agents converge on overlapping subsets of targets. The most pronounced gap appears on Football, where every baseline collapses to near-zero completion, while our framework remains close to 100\%. We attribute this to the task's reliance on event-driven role transitions---possession changes, ball proximity, opponent positioning---which static or scheduled allocations cannot track.

\begin{figure}[t]
    \centering
    \begin{tabular}{cccc}
    \begin{subfigure}[b]{0.256\textwidth}
        \centering
        \captionsetup{margin={0.6cm, 0cm}}
        \includegraphics[width=\linewidth]{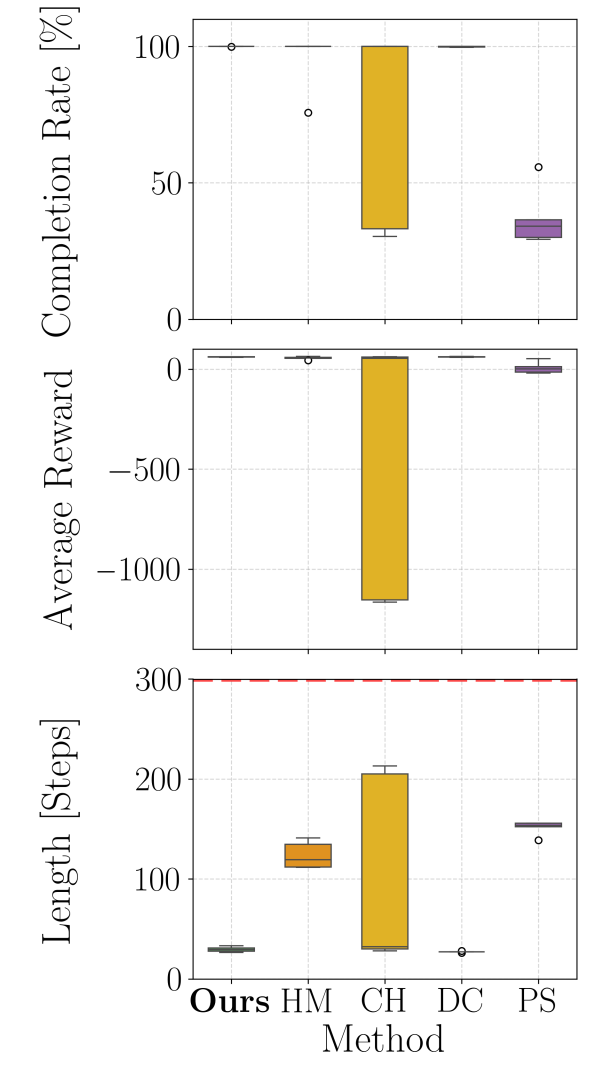}
        \caption{Reverse Transport}
        \label{fig:result_reverse_transport}
    \end{subfigure}
    \begin{subfigure}[b]{0.213\textwidth}
        \centering
        \includegraphics[width=\linewidth]{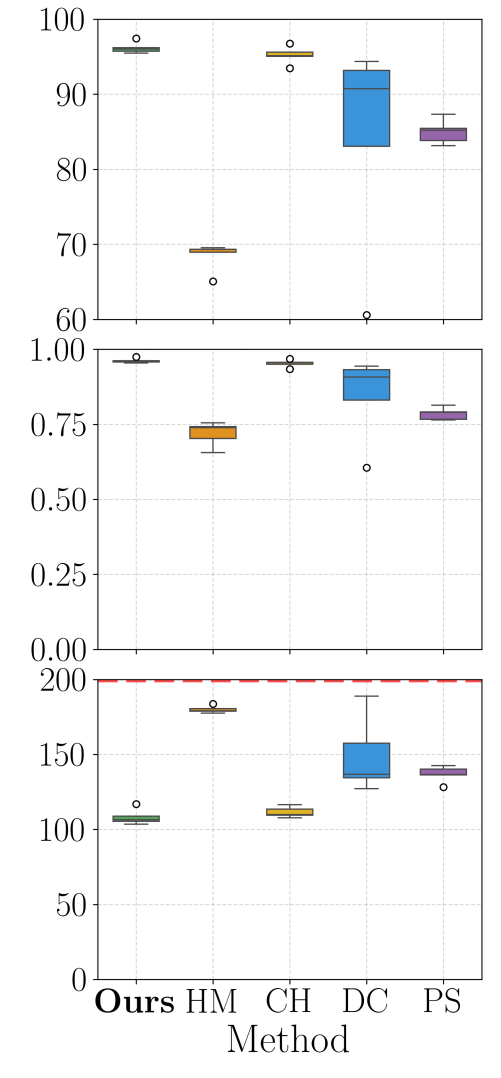} 
        \caption{Navigation}
        \label{fig:result_navigation}
    \end{subfigure}
    \begin{subfigure}[b]{0.213\textwidth}
        \centering
        \includegraphics[width=\linewidth]{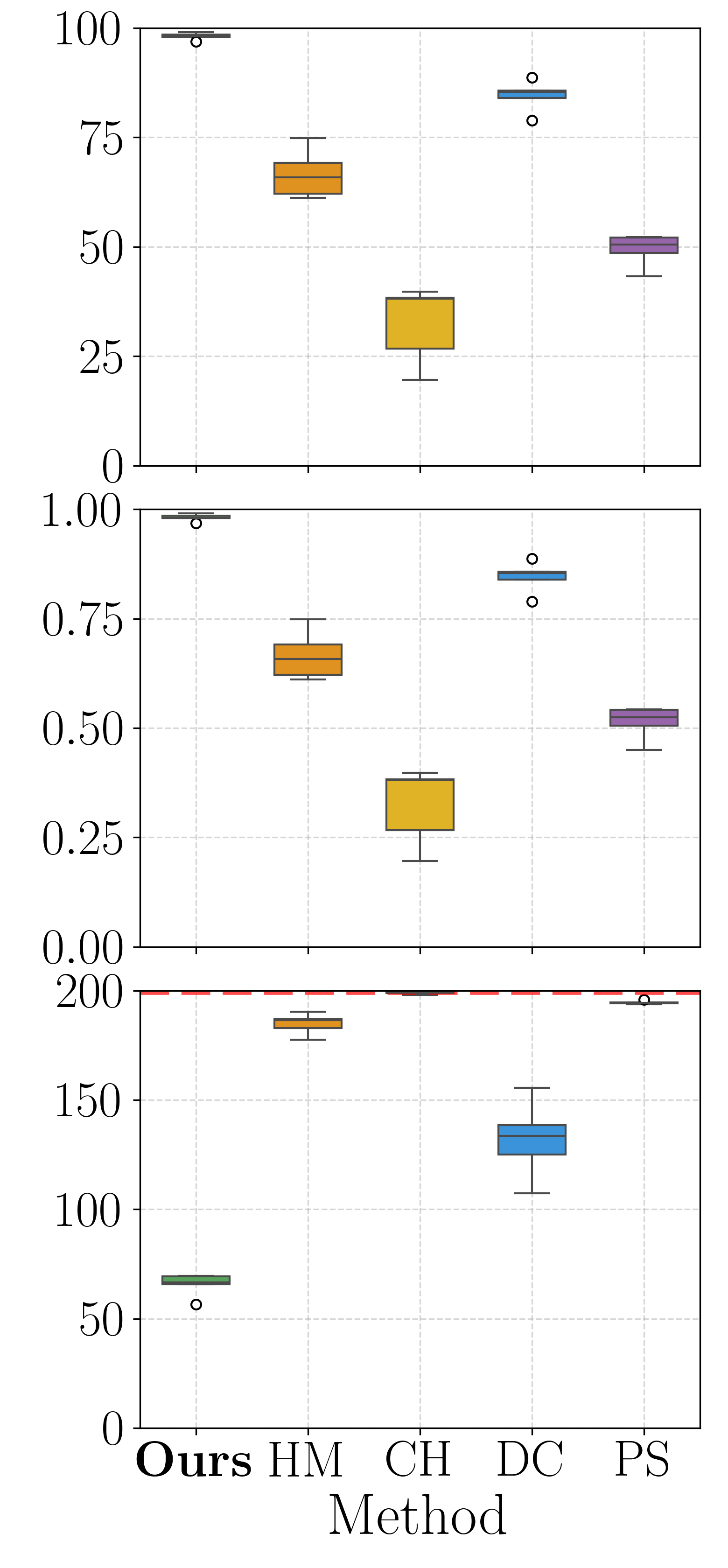}
        \caption{Dispersion} 
        \label{fig:result_dispersion}
    \end{subfigure}
    \begin{subfigure}[b]{0.213\textwidth}
        \centering
        \includegraphics[width=\linewidth]{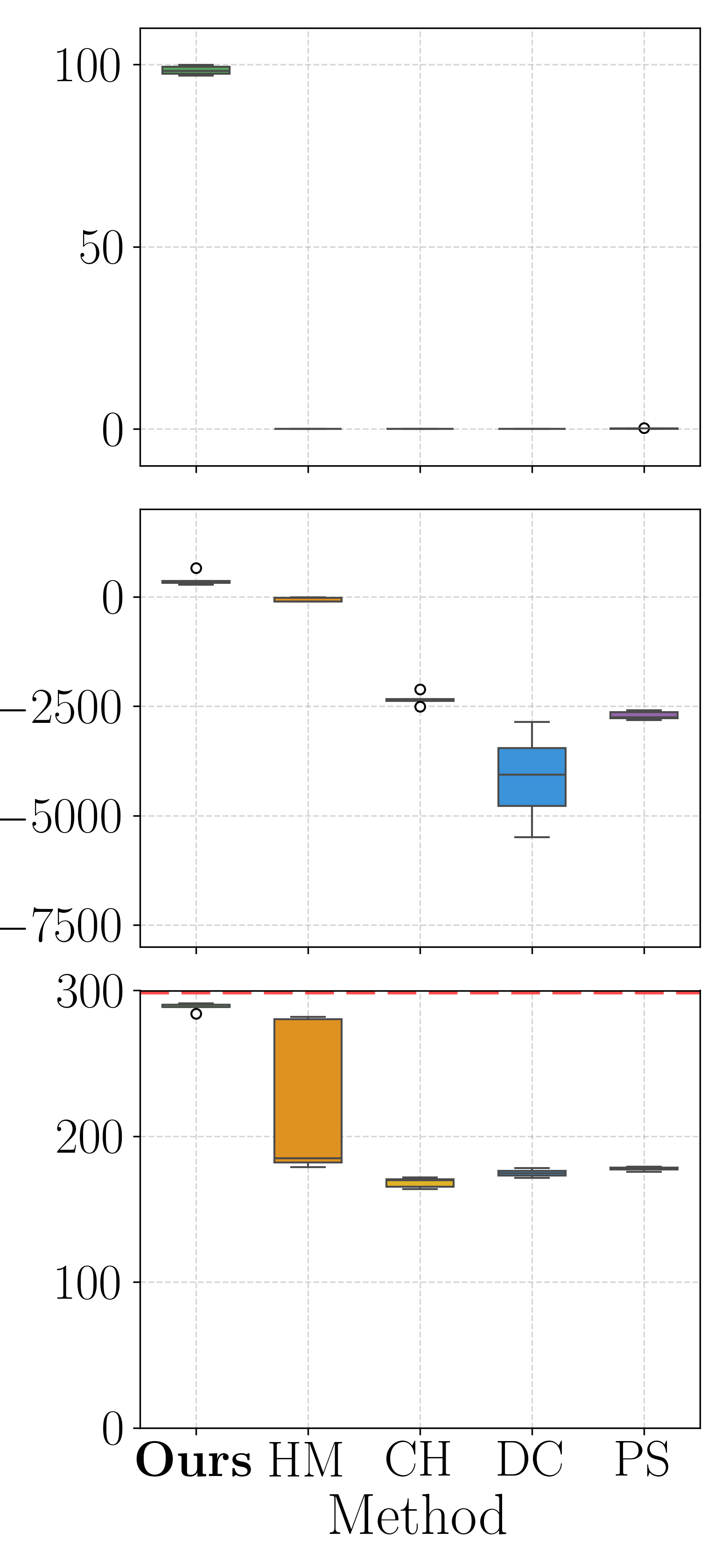}
        \caption{Football}
        \label{fig:result_football}
    \end{subfigure}
    \end{tabular}
    \caption{
    \textbf{Performance Comparison.}
    Final completion rate, average reward, and episode length across $5$ seeds. 
    Colors represent methods: green (Ours), orange (HM: HyperMARL), yellow (CH: CASH), blue (DC: DiCo), purple (PS: Parameter Sharing).
    }
    \label{fig:four_tasks_row}
\end{figure}

To understand whether the performance gains come from the event-trigger scheme, \metric{}-based diversity control, or the LoRA architecture, we ablate each component while holding the others fixed on Navigation, Dispersion and Reverse Transport (Table~\ref{tab:ablation}, Appendix~\ref{app:ablation}). The conclusion is that no single component explains all the performance gains. The full method attains the highest median completion rate on every task, every variant lowers the median on at least one task, and every variant markedly increases seed-to-seed variability. The principal effect of each component is thus to make training reliable rather than to raise peak performance. Further details on the ablation study can be found in Appendix~\ref{app:ablation}.


\subsection{Q2: Zero-shot Generalization across Number of Agents, Capabilities and Target Diversity}

To address \textbf{(Q2)}, we evaluate generalization across three axes. For agent count generalization, we exclude \textbf{HyperMARL} as its architecture cannot accommodate changes in team size, and replicate learned policies for \textbf{DiCo} and \textbf{PS} to enable comparison. For capability generalization, the considered tasks are homogeneous by default, so we run independent training runs for multiple capability values, concatenating capabilities to observations as inputs to the hypernetworks for our framework and \textbf{CASH}. The varied capabilities are: LiDAR range for Navigation, agent speed for Dispersion and Football, and agent force for Reverse Transport. Diversity target generalization is enabled directly by our constrained formulation and requires no additional training runs.

The results in Fig.~\ref{fig:scaling_results} show that our proposed approach maintains performance across these axes of zero-shot generalization, while baselines either stagnate or lose performance. Along the agent-count axis, our framework remains stable across the entire in-distribution and out-distribution ranges for all tasks, evidence that the hypernetwork has learned a behavior manifold in which changing the number of agents simply implies sampling a different amount of points in $\mathcal{B}$. The same conclusion can be drawn from the capability axis. Variations in LiDAR range, agent speed, and pushing force leave the high completion rate of our method largely unchanged, consistently outperforming baselines.

Along the target \metric{} axis, our completion rate is stable across roughly two orders of magnitude of \metric{} values. Interestingly, the absolute value of these ranges vary across tasks, which implies that there is actually a ``good'' region of \metric{} values. The stability of this last axis is the empirical counterpart to Theorem \ref{theorem:joint}: because diversity is enforced through a projection, the achievable reward does not change as $\mathrm{\metric{}}_{\mathrm{des}}$ is varied within reasonable ranges. Overall, this also helps to explain the success of our approach, as diversity control maintains a sufficient level of diversity to span all necessary behaviors to solve the task. Appendix \ref{app:nmd_snd} collects this and further evidence that \metric{} controls functionally meaningful differences.

\begin{figure}[t]
    \centering
    
    \begin{subfigure}{0.245\textwidth}
        \includegraphics[width=\linewidth]{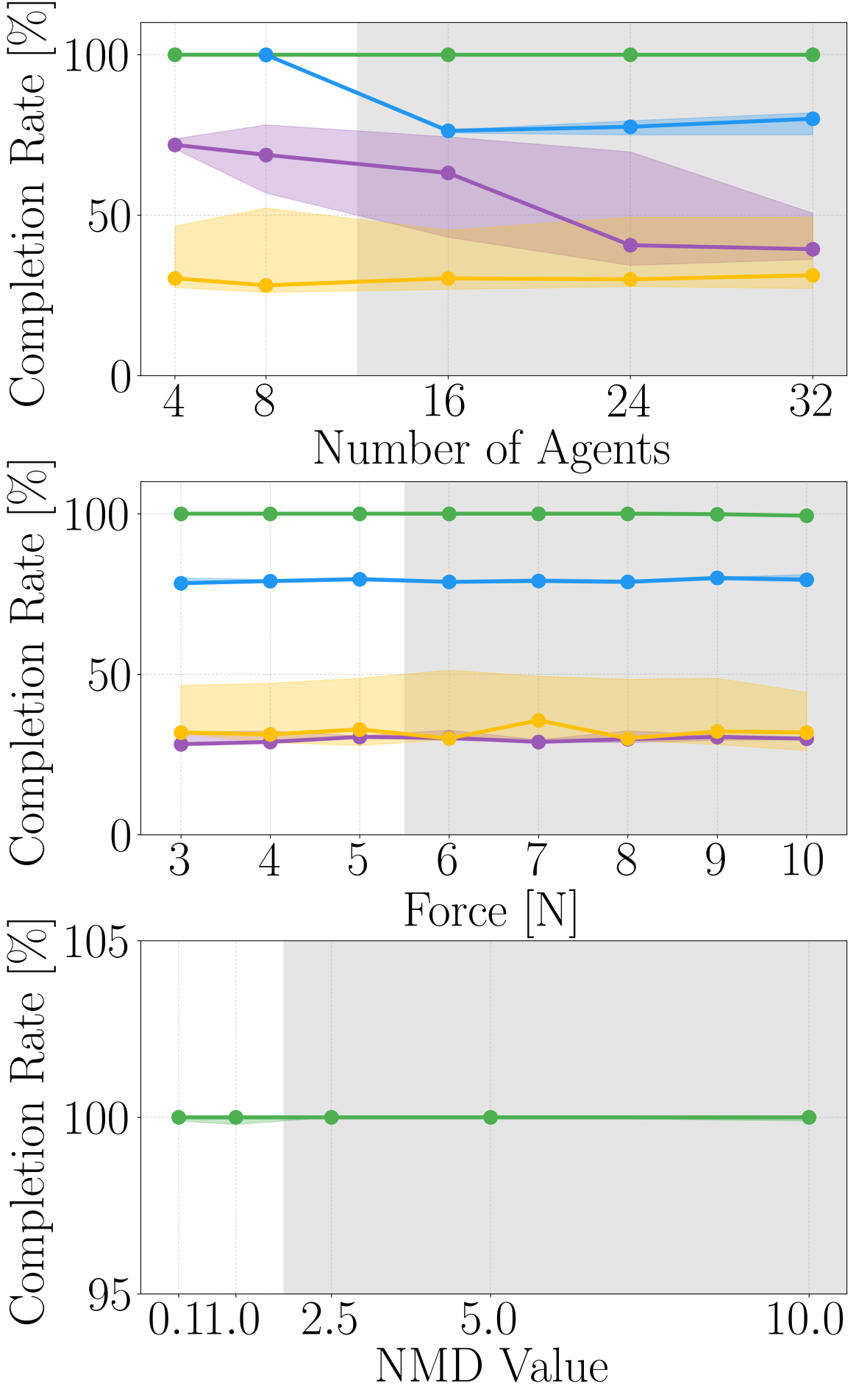} 
        \caption{Reverse Transport}
        \label{fig:nav}
    \end{subfigure}\hfill
    \begin{subfigure}{0.233\textwidth}
        \includegraphics[width=\linewidth]{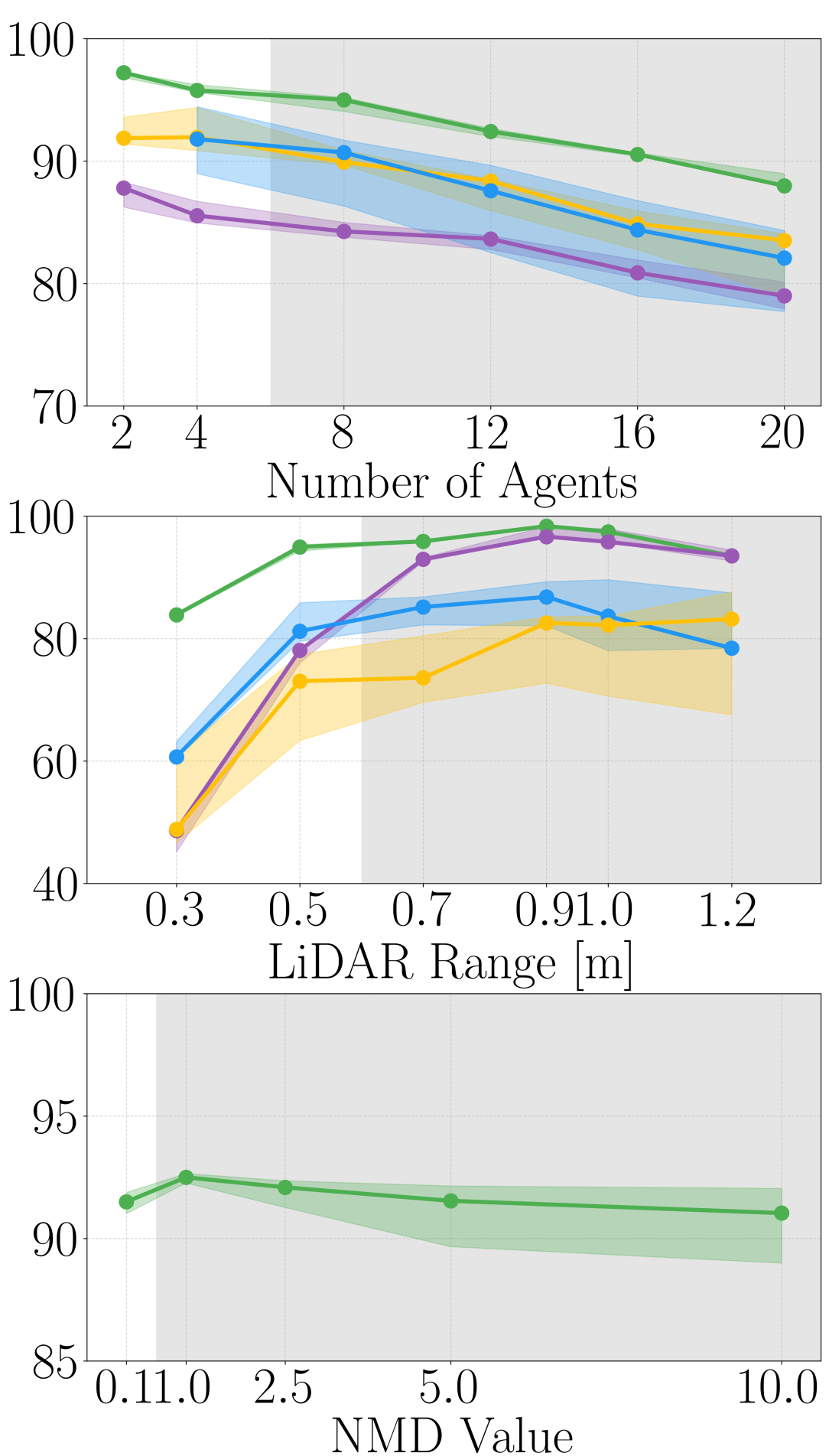}
        \caption{Navigation}
        \label{fig:disp}
    \end{subfigure}\hfill
    \begin{subfigure}{0.231\textwidth}
        \includegraphics[width=\linewidth]{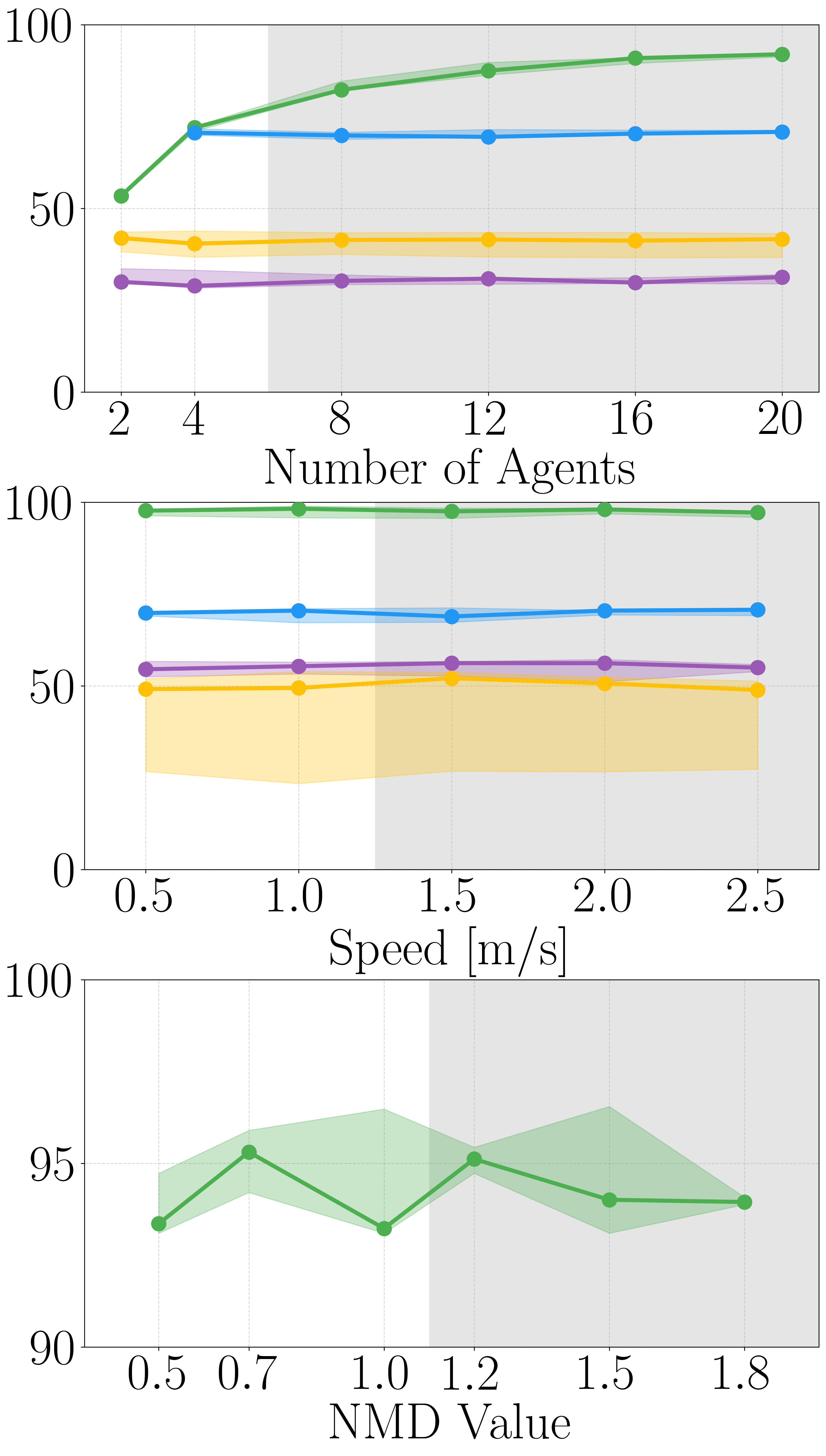}
        \caption{Dispersion}
        \label{fig:rt}
    \end{subfigure}\hfill
    \begin{subfigure}{0.233\textwidth}
        \includegraphics[width=\linewidth]{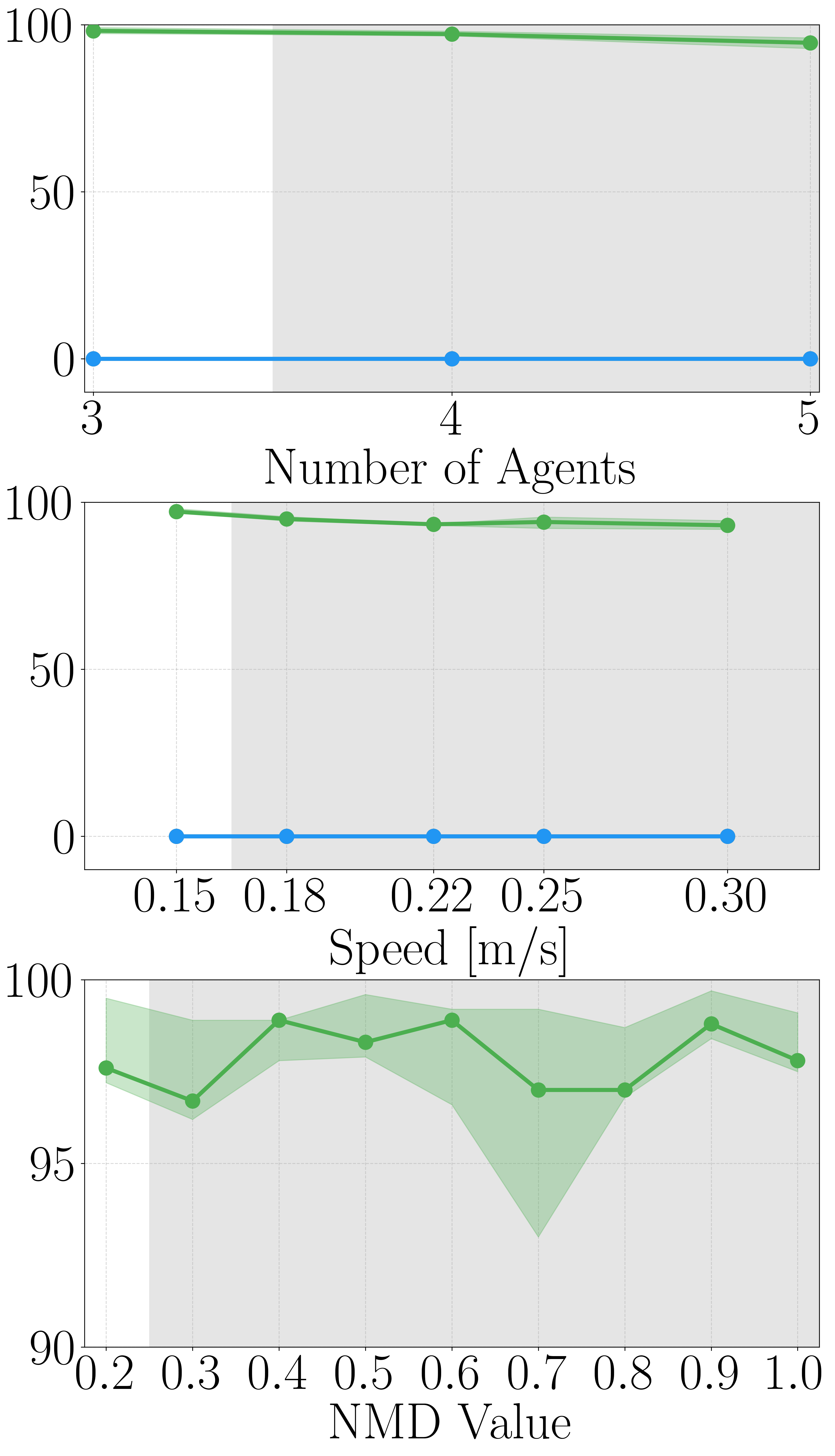}
        \caption{Football}
        \label{fig:foot}
    \end{subfigure}
    
    \caption{
    \textbf{Generalization.}
    Final completion rate across $5$ seeds for the three generalization axes. 
    White areas denote training distribution; gray areas denote zero-shot evaluation. 
    Colors represent methods: green (Ours), yellow (CH: CASH), blue (DC: DiCo), purple (PS: Parameter Sharing).
    }
    \label{fig:scaling_results}
\end{figure}
\subsection{Q3: Adaptive Behaviors in Response to Events}\label{subsec:adapt}

To address \textbf{(Q3)}, we introduce the Pressure Plate task, which requires accurate event-based behavior allocation to succeed. In this scenario, agents must adopt multiple behaviors sequentially; for example, one agent must hold a pressure plate until teammates pass through a door, then transition to a new behavior once a second plate is secured on the opposite side. The results in Fig.~\ref{fig:results_completion} and~\ref{fig:results_reward} show that our approach is the only evaluated method that successfully completes the task.

To verify that this performance is driven by the event mechanism, we compare our approach to an ablation without event re-querying. This ablation fails because agents are statically assigned a single behavior, leaving them unable to adapt to the sequential requirements of the task, which results in the final agent remaining stuck behind the door. Interestingly, our event-based formulation enables sequential strategies without explicit time awareness, an emergent property from the decoupling of agent identity and behavior. This is further studied in Fig. \ref{fig:main}, that visualizes agent trajectories and a t-SNE projection of the behavior space $\mathcal{B}$. Across multiple stochastic runs, distinct clusters emerge, demonstrating that event-based re-querying induces specialized and consistent behavioral modes.

Remarkably, Pressure Plate is a partially observable task: which phase of the task is in progress cannot be recovered from a single observation, so a policy conditioned only on the current observation requires some type of memory. In our framework, the adapters $(C_m, D_m)$ computed at the last event act as a learned, event-indexed register, so that $(o_t, (C_m, D_m)_t)$ is Markov even where $o_t$ alone is not, at the cost of a fixed number of parameters rather than a recurrent network. A sufficiently expressive history-conditioned policy could in principle represent the same information; our claim is therefore about optimization and stability rather than representational capacity. However, in our framework, holding a behavior fixed between events provides temporally extended commitment, analogous to options in single-agent RL~\citep{sutton1999between}, which a per-step policy must otherwise learn to reproduce. On the other hand, CASH, which regenerates its policy at every step, is the least stable method in our comparison for that same reason (Appendix~\ref{app:meanreward}).

\begin{figure*}[t]
    \centering
    
    
    \begin{subfigure}[b]{0.32\textwidth}
        \centering
        \includegraphics[width=\linewidth]{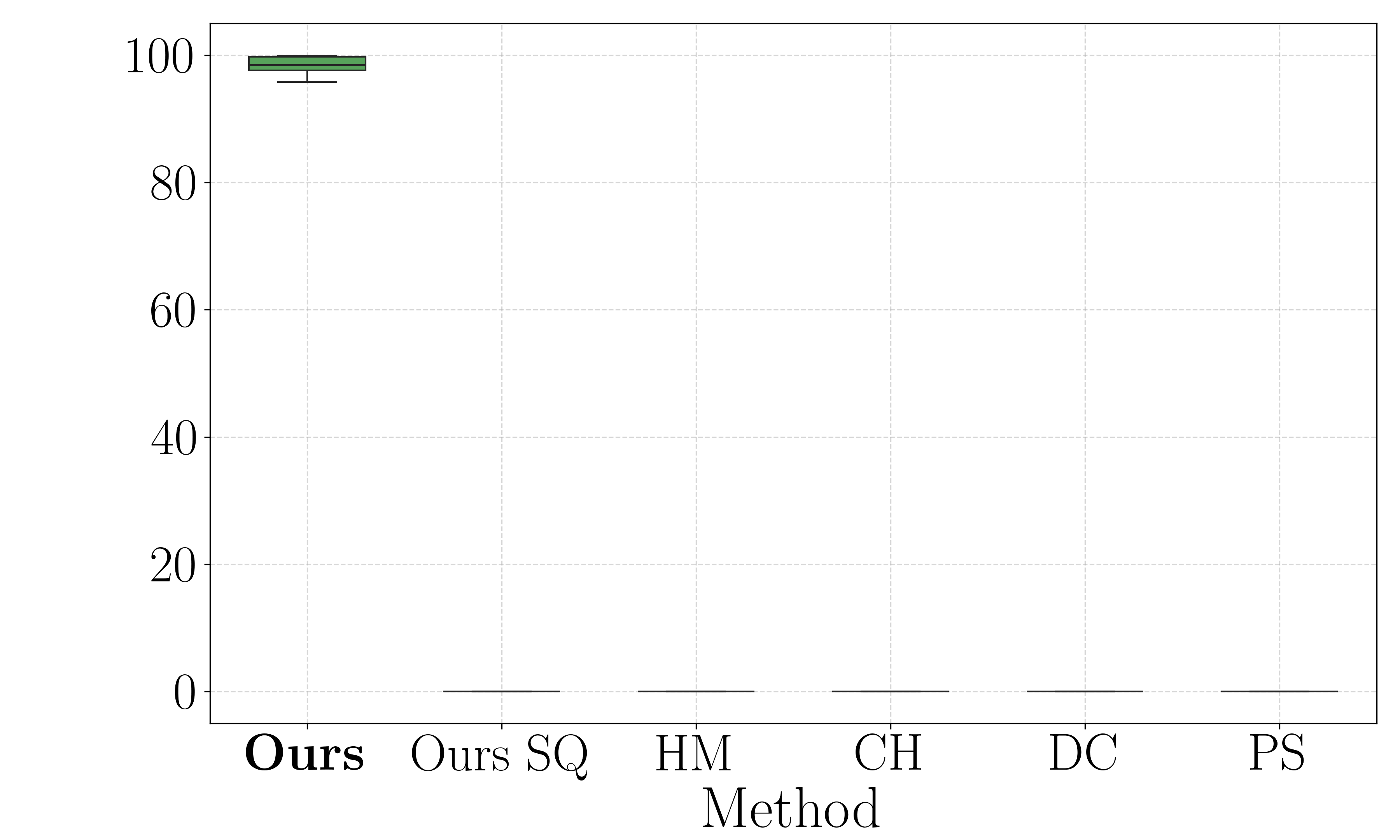}
        \caption{Pressure Plate - Completion {$[\%]$}}
        \label{fig:results_completion}
    \end{subfigure}\hfill
    \begin{subfigure}[b]{0.32\textwidth}
        \centering
        \includegraphics[width=\linewidth]{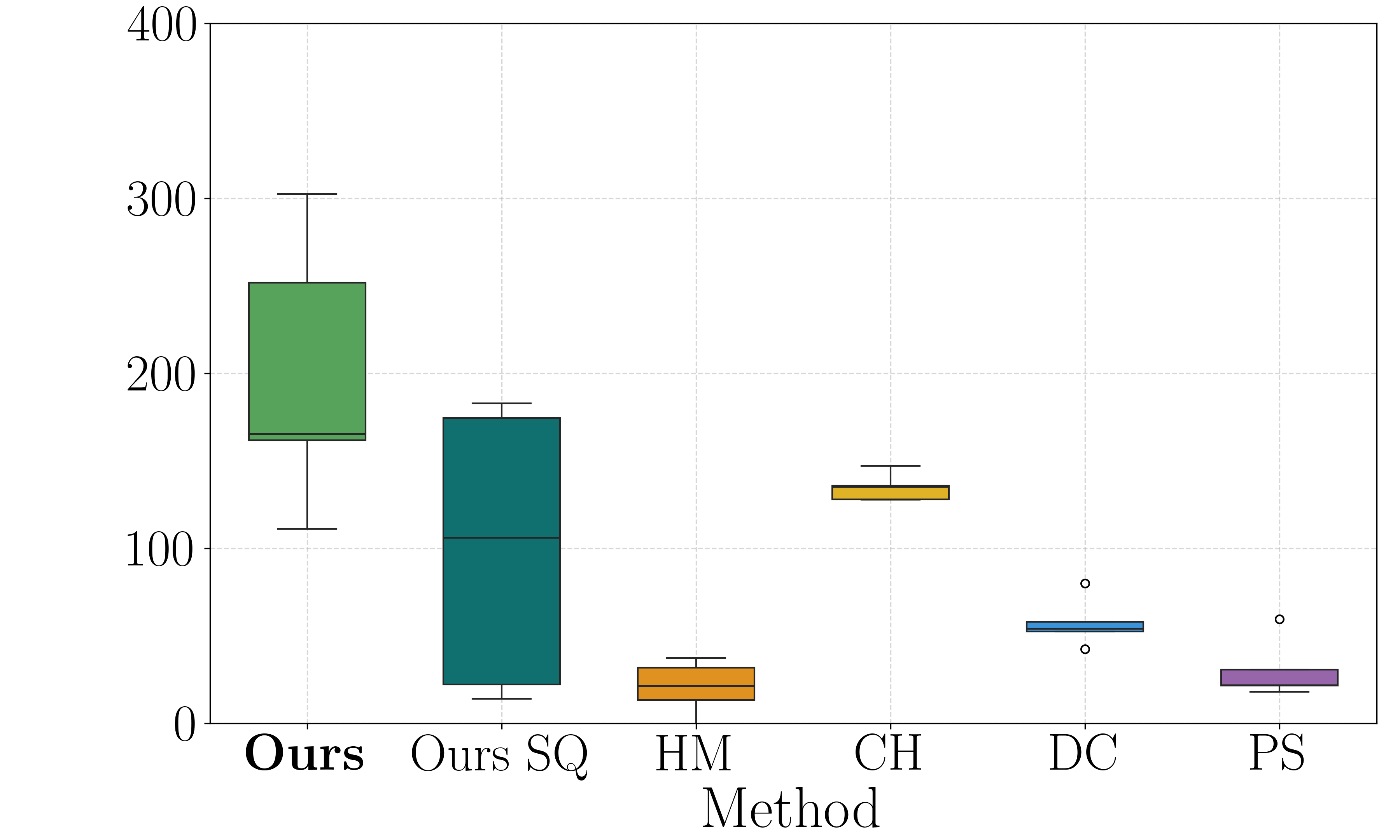}
        \caption{Pressure Plate - Reward}
        \label{fig:results_reward}
    \end{subfigure}\hfill
    \begin{subfigure}[b]{0.32\textwidth}
        \centering
        \vspace{0pt} 
        \resizebox{0.9\linewidth}{!}{
        \begin{tabular}{llc}
            \toprule
            \textbf{Environment} & \textbf{Agent Removal} & \textbf{Completion Rate} \\ 
            \addlinespace[0.5em]
            & & \textbf{Median$^{+\text{(Q3-M)}}_{-\text{(M-Q1)}}$ (\%)} \\
            \addlinespace[0.5em]
            \midrule
            Pressure Plate & No  & $98.4^{+1.4}_{-0.8}$ \\
            \addlinespace[0.5em]
            Pressure Plate & Yes & $81.2^{+0.9}_{-0.2}$ \\
            \midrule
            Football       & No  & $98.2^{+1.0}_{-0.8}$ \\
            \addlinespace[0.5em]
            Football       & Yes & $88.6^{+6.6}_{-0.7}$ \\
            \bottomrule
        \end{tabular}%
        }
        \vspace{0.4cm} 
        \caption{Agent Removal Event}
        \label{tab:combined_agent_removal}
    \end{subfigure}
    
    
    \begin{subfigure}[b]{\textwidth}
        \centering
        \includegraphics[width=0.21\textwidth]{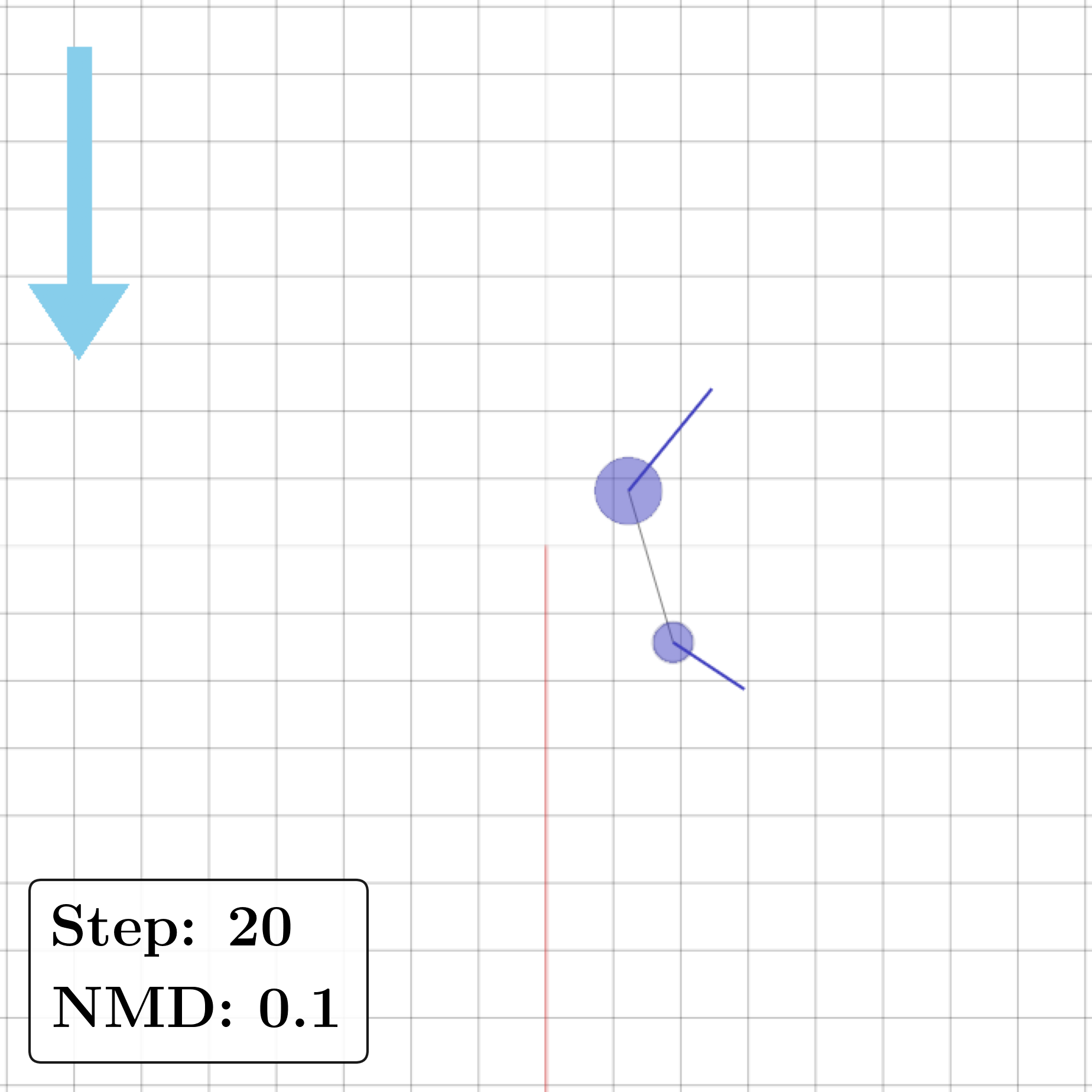}\hfill
        \includegraphics[width=0.21\textwidth]{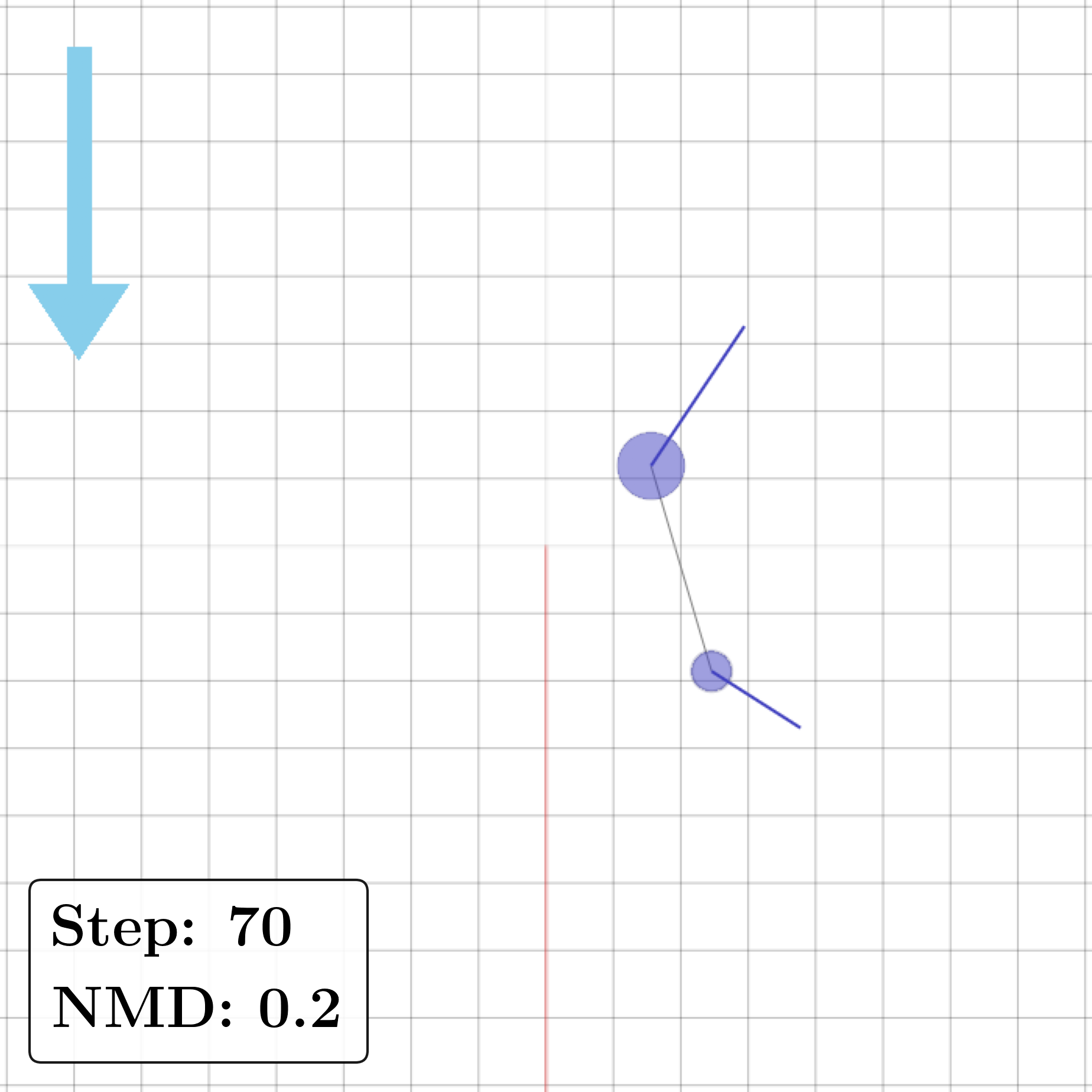}\hfill
        \includegraphics[width=0.21\textwidth]{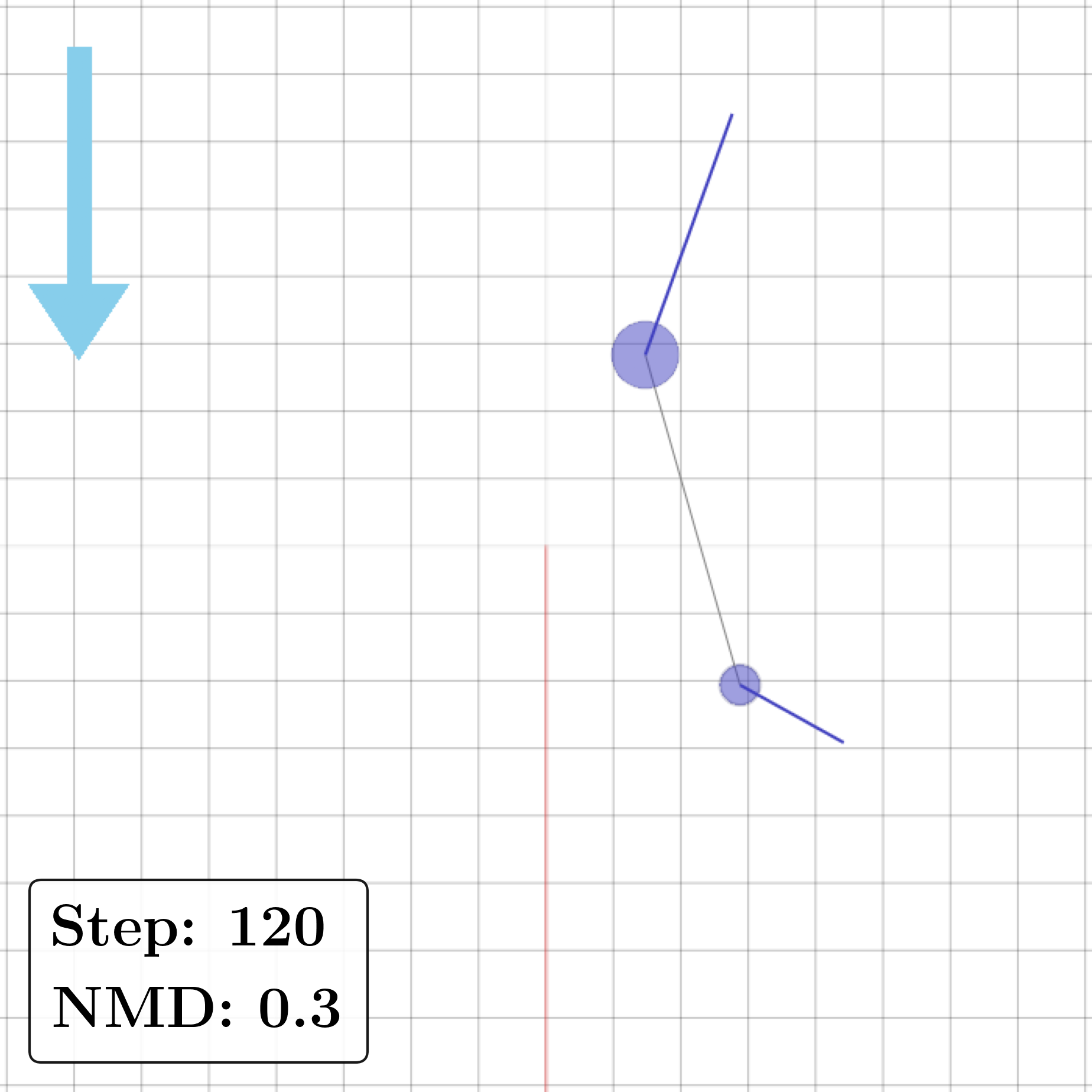}\hfill
        \includegraphics[width=0.21\textwidth]{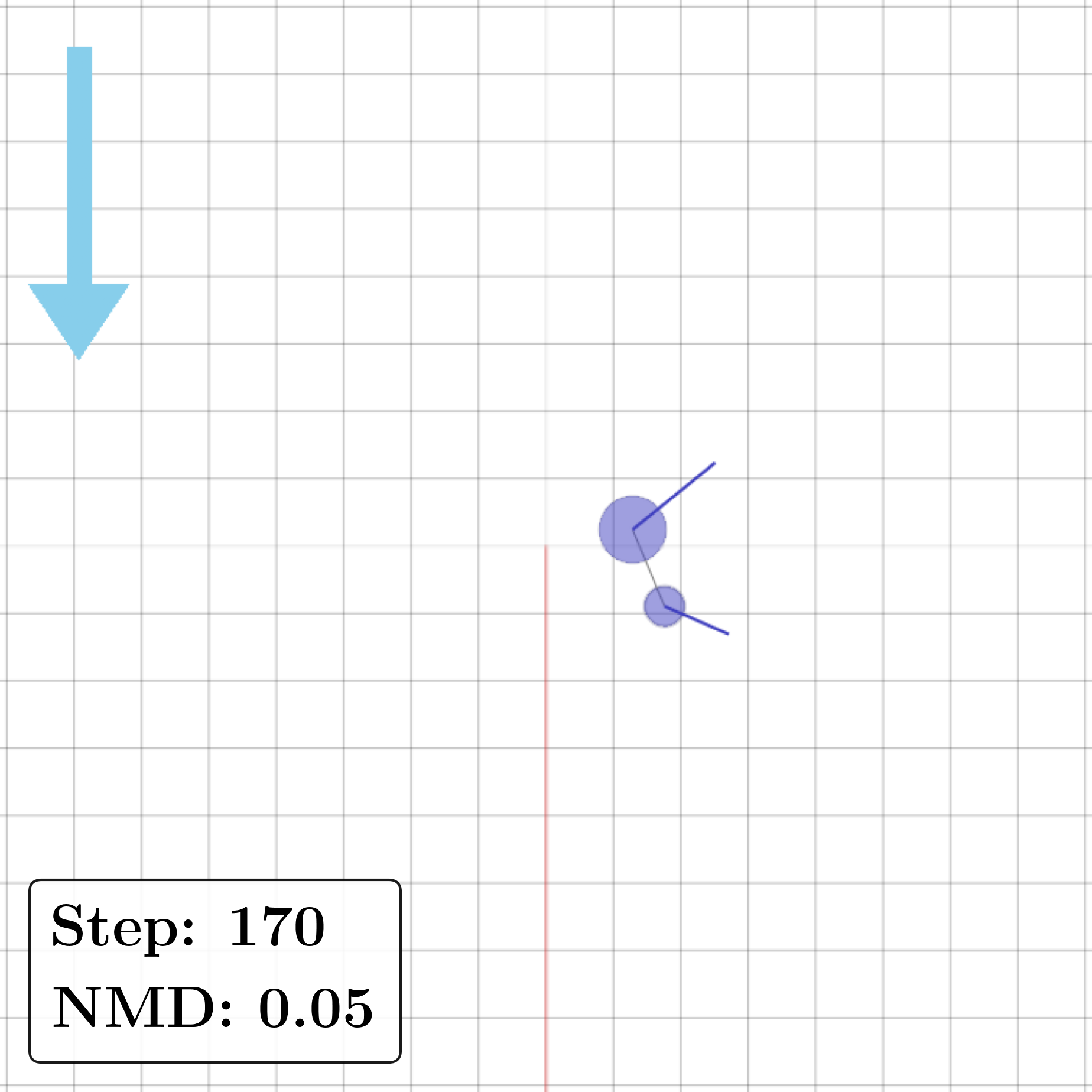}
        \caption{Qualitative visualization in Wind Flocking at steps 20, 70, 120, and 170.}
        \label{fig:windflock_sequence}
    \end{subfigure}
    
    \caption{
    \textbf{Study of Events.} \textbf{(a)} Completion rate and \textbf{(b)} average episode reward on Pressure Plate.
    Colors: green (Ours and SQ: Single Query), orange (HM: HyperMARL), yellow (CH: CASH), blue (DC: DiCo), purple (PS: Parameter Sharing).
    \textbf{(c)} Agent removal on Pressure Plate and Football. 
    \textbf{(d)} Visualization for four \metric{} values in Wind Flocking. The arrow indicates wind direction. 
    }
\end{figure*}

\subsection{Q4: Generalization to Unseen Event Sequences}

To address \textbf{(Q4)}, we evaluate the Football, Pressure Plate, and Wind Flocking environments by introducing unseen event sequences. Specifically, in Football, we remove an agent at a random point; in Pressure Plate, we remove the agent that first reaches the second pressure plate; and in Wind Flocking, we randomly change the target diversity mid-episode. Results in Table~\ref{tab:combined_agent_removal} demonstrate that our method is robust to these additional events and successfully recovers, highlighting its ability to assign relevant behaviors following such events. Specifically, neither perturbation pushes performance into the regime occupied by the baselines in Section \ref{subsec:q1}, suggesting that the framework's response to unseen events is qualitatively close to redistributing roles.

Fig.~\ref{fig:windflock_sequence} visualizes the impact of diversity on the Wind Flocking environment: as the requested diversity decreases, the agents form tighter flocks, highlighting the direct effect of \metric{} on agent behavior. The four snapshots correspond to \metric{} targets that change mid-episode. The agent geometry tracks each change visibly: the inter-agent distance grows monotonically with \metric{} between steps $20$ and $120$, and contracts again at step $170$ once \metric{} is reduced below the initial value. Because the target diversity is supplied to the hypernetwork as an input scalar, this transition does not require retraining and confirms that $\mathrm{\metric{}}_{\mathrm{des}}$ acts as a direct, interpretable control on the realized formation rather than as a soft regularization target. Quantitative results are reported in Appendix~\ref{app:meanreward}.

\section{Conclusion, Limitations, and Future Work}
\label{sec:conc}

This work argues that behavioral adaptation in MARL is best understood as event-driven: triggered by task-relevant transitions rather than by agent identity or fixed temporal cadence. From this premise follows a structural decoupling in which behaviors are no longer tied to agents but instantiated from a continuous manifold the team shares. As a consequence, the meaning of diversity is shifted to a property of the task, not of the agents that populate it. We made this concrete through \longmetric{}, which we proved is a valid distance on behaviors, and through an event-driven hypernetwork that generates LoRA modules over a shared policy. We also proved that this implementation aligns reward maximization with diversity preservation by design. Our empirical results demonstrated that our framework facilitates more effective coordination and generalization compared to existing approaches, outperforming baselines. The benefits of an event-based formulation are further validated in a pressure plate task, in which sequential assignment emerges to solve the task despite the absence of a memory mechanism (e.g., RNN) that learns a state-machine representation of the task status.

Although this work establishes the foundation for an event-based view of behavioral adaptation in MARL, it presents limitations. The societal implications of these are discussed in Appendix \ref{app:societal_impact}, while here we focus on the technical limitations. First, the expressivity of behaviors is inherently constrained by the choice of representation. We use an implicit behavior typing representation where behaviors are inferred from local observations and attention to teammate observations via a hypernetwork, which results in diversity expressed within the space of policy parameters. Consequently, the richness of the strategy space is tied to this specific representation. As an initial effort to investigate this constraint, we provide experiments regarding the dimensionality of our LoRA adapter in Appendix~\ref{app:lorarank}. In continuous-control domains with redundant actuation, large differences in action space need not translate into different state transitions; Appendix \ref{app:nmd_snd} discusses the evidence we have on this point and the extension to trajectory-level distances. Second, our results hold under any metric on a discrete action set, so NMD remains well defined there, but the exact scaling of Eq. \eqref{eq:alpha} does not, because the softmax is not a linear operation; the constraint would hold only approximately. Third, we note that in our current formulation, events are defined a priori using domain knowledge. A compelling extension for future work involves inferring these events either end-to-end or through latent discovery. This logic similarly applies to the diversity target, as our method generalizes across a range of \metric{} values, yet the identification of optimal values for this metric could likely be automated through an additional optimization objective. Lastly, the realization of our framework through a hypernetwork is currently centralized, requiring access to the observations from all agents to assign behaviors. It would be interesting to derive an extension of our solution that is decentralized, in which agents decide on their behavior based only on neighboring information.

\begin{ack}
This work is supported by the European Research Council (ERC) Project 949940 (gAIa), by a Leverhulme Trust Research Project Grant, and by the EPSRC funded INFORMED-AI project EP/Y028732/1. We gratefully acknowledge their support.
\end{ack}

\bibliographystyle{plainnat} 
\bibliography{references}    

\appendix

\section{Extended Related Work} \label{app:relatedwork}

\textbf{Temporal Behavioral Assignment.}
A core question for behavior allocation is the assignment timeline. Existing approaches fall into two categories. The first one comprises methods that fix one behavior per agent for the entire episode~\citep{Tessera2024, bettini2024controlling, jiang2021emergence, li2021celebrating, liu2025drol, christianos2021seps}, either because they learn per-agent policies~\citep{bettini2024controlling, jiang2021emergence, christianos2021seps}, or because they do not allow dynamically changing behavior~\citep{Tessera2024, li2021celebrating, liu2025drol}. 
The second category comprises methods that adapt behaviors on a fixed temporal schedule---either every timestep or every $c$ steps---regardless of task conditions~\citep{Fu2025, qi2026rois, mahajan2019maven, wang2020roma, wang2020rode, goel2025r3dm}.

\paragraph{Representations of Behaviors.}
Behavioral specialization is fundamentally constrained by the underlying behavior representation. 
A common paradigm is to learn explicit policies for each behavior, with methodologies ranging from learning deviations from a common policy~\citep{li2021celebrating, bettini2024controlling} to selective parameter sharing~\citep{christianos2021seps, kim2023networkpruning, li2024adaptive, jiang2021emergence}. However, these approaches do not easily accommodate dynamic behavior allocation as the learned behaviors remain fixed and tied to specific agents. 
Another frequent approach is to use a shared policy conditioned on agent identity, typically defined as one-hot encoding~\citep{Tessera2024}, or capabilities when available~\citep{Fu2025}. This method, also known as behavior typing~\citep{bettini2023heterogeneous}, has been applied through both direct policy conditioning~\citep{deka2021natural} and conditioned hypernetworks that indirectly predict policy parameters~\citep{Tessera2024, Fu2025}. However, it inherently limits scalability. 
A more flexible extension involves learning a latent behavior space (or role space) from which behaviors are sampled, either to condition a shared policy~\citep{mahajan2019maven, qi2026rois, goel2025r3dm} or a shared hypernetwork~\citep{wang2020roma, liu2025drol}. While this has shown promising results, it relies on training additional encoder networks and requires prior knowledge on the dimensionality and structure of the latent space to identify distinct behaviors. 
Alternatively,~\citet{wang2020rode} achieves specialization by partitioning the action space, though this requires fixing the decomposition early to maintain tractability.
In contrast to these explicitly conditioned or constrained methods, inferred behavior typing~\citep{gupta2017cooperative, bettini2023heterogeneous} derives behavior solely from environmental observations.

\paragraph{Quantifying and Enforcing Behavioral Diversity.}
While the previous section focuses on enabling diversity through representation, this is often insufficient to ensure functional divergence in practice. Consequently, while some works do not explicitly reward diversity~\citep{wang2020rode, christianos2021seps}, many works have proposed various methods to encourage diversity within teams. These methods generally fall into two categories: those that encourage diversity through auxiliary rewards or losses~\citep{wang2020roma, li2021celebrating, mahajan2019maven, jiang2021emergence, liu2025drol, qi2026rois, goel2025r3dm}, and those that enforce diversity as a constraint in the learning process~\citep{bettini2024controlling, tan2023policy}. While the first type of approach is easier to construct and optimize, it offers only weak guarantees on final team diversity. The second type provides a strong guarantee that diversity constraints will be satisfied, although it requires a reference target value to be known a priori. Additionally, researchers have proposed different ways to quantify this diversity. Many papers focus on information-theoretic objectives where the trajectories of a policy should carry maximum information about its behavior~\citep{wang2020roma, li2021celebrating, mahajan2019maven, jiang2021emergence, qi2026rois, goel2025r3dm}, requiring training models to differentiate trajectories. 
Alternatively,~\citep{liu2025drol} proposes to use a contrastive loss to distinguish behavior, which also requires a separate model to differentiate trajectories. 
Finally, \citet{bettini2023system} proposed System Neural Diversity (SND), an explicit diversity metric based on the action distribution of the policies.

\paragraph{Event-Triggered Decision Making, Dynamic Role Assignment and Task Allocation.}
Event-based coordination has a long history outside the MARL diversity literature. In event-triggered and self-triggered control~\citep{tabuada2007event, heemels2012introduction, dimarogonas2012distributed, girard2015dynamic, nowzari2019event}, a control law is given and the trigger decides when to re-sample, transmit or update it, with the objective of preserving stability while reducing communication or computation; hence, the behavior itself is fixed by design. Event-triggered reinforcement learning and adaptive dynamic programming~\citep{vamvoudakis2018model, baumann2018deep, yang2019adaptive} learn a triggering condition, or a controller under a triggering condition, for resource efficiency in single-agent regulation problems, so the notion of diversity in action space and behaviors is not applicable. In event-triggered communication for MARL~\citep{hu2023etcnet, mao2020learning, zhang2020succinct}, events gate messages under a bandwidth constraint; the policy set is unchanged by the event, and what changes is which agent receives what information. Dynamic role assignment and multi-robot task allocation~\citep{alliance1998, choi2009consensus, gerkey2004formal, gerkey2004role, dias2006market, korsah2013comprehensive} are genuinely state-dependent and dynamic, but the roles or tasks form a discrete, pre-specified set, the assignment is solved as a combinatorial problem over that set, and the utilities are typically hand-designed, so that what is allocated is defined before learning begins.
In our setting, by contrast, the event decides which behavior is instantiated from a continuous behavior manifold that is itself learned and characterized by \metric{}, so that its diversity can be measured and constrained. This is what enables the generalization results of Section~\ref{sec:experiments}: changing the number of agents amounts to drawing a different number of behaviors from the manifold, and unseen event sequences amount to unseen traversals of the manifold.


\section{Proofs of Theoretical Results}\label{app:proofs}

\subsection{Notation and Assumptions}

\textbf{(A1) Behavior manifold.} The behavior manifold $\mathcal{B}$ is a measurable space of conditional action distributions $\pi : \mathcal{O} \to \mathcal{P}(\mathcal{A})$, where $\mathcal{P}(\mathcal{A})$ denotes the set of Borel probability measures on the action space $\mathcal{A}$. We assume $\mathcal{A} \subseteq \mathbb{R}^{d_a}$ is equipped with the Euclidean norm $\|\cdot\|$, and that for every $\pi \in \mathcal{B}$ and every $o \in \mathcal{O}$, the action distribution $\pi(o)$ has finite second moment.

\textbf{(A2) Observation distribution.} Observations $o \in \mathcal{O}$ are drawn from a probability density $p(o)$, and behaviors $\pi \in \mathcal{B}$ are drawn from a density $p(\pi)$. All integrals over $\mathcal{O}$ and $\mathcal{B}$ are taken with respect to these densities.

\textbf{(A3) Location-family policy parameterization.} When invoked (Lemma \ref{lemma:w2} onward), each behavior $\pi_m$ belongs to a location family $\pi_m(\cdot \mid o) = \mu_m(o) + \varepsilon$, where the noise $\varepsilon$ has a distribution $Q$ that is shared across behaviors and independent of $o$, with finite second moment. The shared-covariance Gaussian $\pi_m(\cdot | o) = \mathcal{N}(\mu_m(o), \Sigma)$ is the instance used in our implementation, where the LoRA update enters only the mean and the log-standard deviation is state-independent and shared.

\textbf{(A4) Linear LoRA parameterization.} The mean of behavior $m$ is decomposed as \mbox{$\mu_m(o) = W_{\mathrm{shared}}\, \phi(o) + \alpha\, u_m(o)$, $u_m(o) := D_m C_m\, \phi(o)$}, where $\phi(o) \in \mathbb{R}^d$ is a shared feature, $W_{\mathrm{shared}} \in \mathbb{R}^{d_a \times d}$ is the shared backbone, $C_m \in \mathbb{R}^{r \times d}$ and $D_m \in \mathbb{R}^{d_a \times r}$ are the LoRA factors, and $\alpha \in \mathbb{R}_{>0}$ is the diversity scalar. The deviation $u_m(o)$ is deterministic.

\textbf{(A5) Differentiability.} The expected return $R$ is differentiable in the LoRA parameters $\{C_m, D_m\}$ and in the shared parameters via $\phi$ and $W_{\mathrm{shared}}$, and the empirical \metric{} estimator $\widehat{\mathrm{\metric{}}}$ defined in Eq. \eqref{eq:nmd_approx} of the main text is differentiable in the deviations $\{u_m\}$ on the open set $\widehat{\mathrm{\metric{}}} > 0$. For a fixed observation $o$, we abbreviate $u_m := u_m(o)$, $\mu_m := \mu_m(o)$, $\phi := \phi(o)$, and $z_m := \mu_m$ (the policy mean coincides with the pre-activation in (A3)–(A4)).

\subsection{Proof of Proposition~\ref{prop:nmd_distance}}
\begin{proof}
We verify the four metric axioms.
    
\textbf{Non-negativity.} For every $o$, the 2-Wasserstein distance satisfies $W_2(\pi_m(o), \pi_n(o)) \geq 0$ (it is a metric on $\mathcal{P}_2(\mathcal{A})$ by \citep{villani2009optimal} Thm. 6.18). The expectation of a non-negative integrand is non-negative, so $d(\pi_m, \pi_n) \geq 0$.

\textbf{Symmetry.} $W_2$ is symmetric in its arguments, hence $W_2(\pi_m(o), \pi_n(o)) = W_2(\pi_n(o), \pi_m(o))$ point-wise in $o$, and the expectation preserves equality.

\textbf{Triangle inequality.} Fix any $\pi_m, \pi_n, \pi_k \in \mathcal{B}$. Point-wise in
$o$, the triangle inequality for $W_2$ gives $$W_2(\pi_m(o), \pi_k(o)) \leq W_2(\pi_m(o), \pi_n(o)) + W_2(\pi_n(o), \pi_k(o)).$$ Taking expectation under $p(o)$ and using linearity and monotonicity of expectation:
$$d(\pi_m, \pi_k) \leq d(\pi_m, \pi_n) + d(\pi_n, \pi_k).$$

\textbf{Identity.} If $\pi_m = \pi_n$ in $\mathcal{B}$, then $\pi_m(o) = \pi_n(o)$ for all $o$, hence $W_2(\pi_m(o), \pi_n(o)) = 0$ point-wise and $d(\pi_m, \pi_n) = 0$. Conversely, $d(\pi_m, \pi_n) = 0$ implies $W_2(\pi_m(o), \pi_n(o)) = 0$ for $p$-almost every $o$, hence $\pi_m(o) = \pi_n(o)$ for $p$-almost every $o$. This gives identity up to $p$-null sets, which establishes the pseudometric property. Under the additional separation assumption the converse is strengthened to $\pi_m = \pi_n$, yielding a genuine metric. 
\end{proof}

\begin{corollary}
Under the separation assumption, $\mathrm{\metric{}}(\mathcal{B}) = 0$ if and only if $p(\pi)$ concentrates on a single behavior. The forward direction follows from $d(\pi_m, \pi_n) = 0 \iff \pi_m = \pi_n$; the converse from non-negativity of $d$ and the integral definition of $\mathrm{\metric{}}$.
\end{corollary}

\subsection{Proof of Theorem \ref{theorem:joint}} 

Before proving the main result of the theorem, we state and prove a Lemma that helps to derive the connection between diversity in the behavior manifold and reward maximization.

\begin{lemma}\label{lemma:w2}
    Let $\pi_m$ and $\pi_n$ be location-family policies with means $\mu_m, \mu_n$ and common noise law $Q$. If the heterogeneous component $u_m(o_t)$ enters $\mu_m$ linearly and is deterministic, then \mbox{$W_2(\pi_m, \pi_n) = \|\mu_m - \mu_n\|_2$}.
\end{lemma}
\begin{proof}
The closed-form expression for the 2-Wasserstein distance between two non-degenerate Gaussians on $\mathbb{R}^{d_a}$ is
$$W_2^2\!\left(\mathcal{N}(\mu_m, \Sigma_m), \mathcal{N}(\mu_n, \Sigma_n)\right) = \|\mu_m - \mu_n\|_2^2 + \mathrm{tr}\!\left(\Sigma_m + \Sigma_n - 2\,(\Sigma_m^{1/2}\Sigma_n \Sigma_m^{1/2})^{1/2}\right),$$ a result due to \cite{olkin1982distance}. Under (A3), $\Sigma_m = \Sigma_n = \Sigma$, so $$\Sigma_m + \Sigma_n - 2(\Sigma_m^{1/2}\Sigma_n \Sigma_m^{1/2})^{1/2} = 2\Sigma - 2(\Sigma^{1/2}\Sigma\,\Sigma^{1/2})^{1/2} = 2\Sigma - 2\Sigma = 0,$$ where we used that for $\Sigma \succ 0$, $\Sigma^{1/2} \Sigma\, \Sigma^{1/2} = \Sigma^2$ and $(\Sigma^2)^{1/2} = \Sigma$. The trace term vanishes, leaving
$$W_2^2(\pi_m(o), \pi_n(o)) = \|\mu_m(o) - \mu_n(o)\|_2^2.$$ Taking square roots (both sides are non-negative) gives the claim. 
\end{proof}

Substituting Lemma 6 into Eq. \eqref{eq:nmd_approx} of the main text and using $\mu_m(o) - \mu_n(o) = \alpha\,(u_m - u_n)$ from (A4), we obtain
\begin{equation}\label{eq:app_1}
\widehat{\mathrm{NMD}}\big(\{\pi_m\}_{m=1}^{B}\big) = \frac{2\alpha}{B(B-1)\,|\mathcal{O}_{\mathrm{ep}}|} \sum_{m=1}^{B}\sum_{n=m+1}^{B}\sum_{o\in\mathcal{O}_{\mathrm{ep}}} \|u_m - u_n\|_2 .
\end{equation}
This makes explicit that $\widehat{\mathrm{\metric{}}}$ is a homogeneous function of degree 1 in the deviations $\{u_m\}$, result that leads to the proof of Theorem \ref{theorem:joint}.

\begin{proof}
We compute $\nabla_{u_m} R$ via the chain rule, treating $\alpha$ as a function of $\{u_n\}_{n=1}^B$ through Eq. \eqref{eq:app_1}.

\textbf{Step 1.} From (A4), $z_m = W_{\mathrm{shared}}\phi + \alpha u_m$. The dependence of $z_m$ on $u_m$ has two contributions: the explicit $\alpha u_m$ term, and the implicit dependence of $\alpha$ on $u_m$ through $\widehat{\mathrm{\metric{}}}$. Treating $u_m$ as a vector in $\mathbb{R}^{d_a}$: $$\frac{\partial z_m}{\partial u_m} = \alpha\, I + u_m \otimes \nabla_{u_m} \alpha,$$ where $\otimes$ denotes the outer product (we use the convention $(a \otimes b)\,v = a\,(b^\top v)$. Substituting $\alpha = \mathrm{\metric{}}_{\mathrm{des}}/\widehat{\mathrm{\metric{}}}$ and applying the quotient rule:
$$\nabla_{u_m} \alpha = -\frac{\mathrm{\metric{}}_{\mathrm{des}}}{\widehat{\mathrm{\metric{}}}^{\,2}}\, \nabla_{u_m} \widehat{\mathrm{\metric{}}} = -\frac{\alpha}{\widehat{\mathrm{\metric{}}}}\, \nabla_{u_m} \widehat{\mathrm{\metric{}}}.$$ Therefore,
$$\frac{\partial z_m}{\partial u_m} = \alpha\, I - \frac{\alpha}{\widehat{\mathrm{\metric{}}}}\, u_m \otimes \nabla_{u_m} \widehat{\mathrm{\metric{}}} = \alpha\, P_{u_m},$$ with $P_{u_m}$ as in Theorem \ref{theorem:joint}.

\textbf{Step 2.} By the chain rule, $$\nabla_{u_m} R = \left(\frac{\partial z_m}{\partial u_m}\right)^\top \nabla_{z_m} R = \alpha\, P_{u_m}^\top\, \nabla_{z_m} R.$$ Because $P_{u_m}$ is symmetric in the dyadic structure that matters for the projection identity below (we verify idempotency on the same form), we drop the transpose for notational convenience and write $$\nabla_{u_m} R = \alpha\, P_{u_m}\, \nabla_{z_m} R.$$ 

\textbf{Step 3.} Using Eq. \eqref{eq:app_1}, $\widehat{\mathrm{\metric{}}}$ is positively homogeneous of degree 1 in $u_m$. By Euler's homogeneous function theorem,
\begin{equation}\label{eq:app_4}
u_m^\top\, \nabla_{u_m} \widehat{\mathrm{\metric{}}} = \widehat{\mathrm{\metric{}}}.
\end{equation}
Now compute $P_{u_m}^2$. Writing $v := \nabla_{u_m} \widehat{\mathrm{\metric{}}}$ and $N := \widehat{\mathrm{\metric{}}}$ for brevity,
$$P_{u_m}^2 = \left(I - \frac{u_m v^\top}{N}\right)\!\left(I - \frac{u_m v^\top}{N}\right) = I - \frac{2\,u_m v^\top}{N} + \frac{u_m (v^\top u_m)\, v^\top}{N^2}.$$
By Eq. \eqref{eq:app_4}, $ v^\top u_m = N$, so the last term simplifies:
$$\frac{u_m (v^\top u_m)\, v^\top}{N^2} = \frac{u_m\, N\, v^\top}{N^2} = \frac{u_m v^\top}{N}.$$
Substituting back,
$$P_{u_m}^2 = I - \frac{2\,u_m v^\top}{N} + \frac{u_m v^\top}{N} = I - \frac{u_m v^\top}{N} = P_{u_m}.$$
Thus $P_{u_m}$ is idempotent.
\end{proof}

\subsection{Proof of Corollary \ref{corollary:limits}}
\begin{proof}
(i) The empirical NMD entering $P_{u_m}$ is the denominator of Eq.~\eqref{eq:alpha}, i.e., the symmetric pairwise sum of Eq.~\eqref{eq:app_1} without the factor $\alpha$,
$$\widehat{\mathrm{NMD}} \propto \sum_{n\neq k} \|u_n - u_k\|_2 \Big/ \big[\text{normalization}\big].$$
Differentiating with respect to $u_m$:
$$\nabla_{u_m}\widehat{\mathrm{NMD}} = c \sum_{n\neq m} \frac{u_m - u_n}{\|u_m - u_n\|_2},$$
where, for a fixed observation $o$, $c = 2/\big(B(B-1)\,|\mathcal{O}_{\mathrm{ep}}|\big)$ is a positive constant depending only on $B$ and $|\mathcal{O}_{\mathrm{ep}}|$; its value plays no role in the limits below.  As $\|u_m\| \to 0$ (with all $\|u_n\|$ for $n \neq m$ bounded away from zero), the summand vectors $(u_m - u_n)/\|u_m - u_n\|_2$ remain bounded in norm by $1$, but the second-moment contribution to $P_{u_m}$ is the rank-one term $u_m \otimes \nabla_{u_m}\widehat{\mathrm{\metric{}}} / \widehat{\mathrm{\metric{}}}$. The numerator factor $u_m \to 0$ drives this rank-one term to zero in operator norm, while the denominator $\widehat{\mathrm{\metric{}}}$ remains bounded below by a positive constant (since the remaining behaviors maintain non-trivial pairwise distances). Hence $P_{u_m} \to I$.

(ii) Write $u_m = t\hat{u}_m$ with $t \to \infty$ and $\hat{u}_m$ fixed. For large $t$, the dominant pairs $(u_m, u_n)$ in $\widehat{\mathrm{\metric{}}}$ are those involving $u_m$, since $\|u_m - u_n\| \sim t$ while $\|u_n - u_k\|$ for $n, k \neq m$ remains $O(1)$. Asymptotically,
$$\widehat{\mathrm{\metric{}}} = \frac{c\,(B-1)}{\,B(B-1)/2\,}\, t + O(1) = \frac{2c}{B}\, t + O(1),$$
and
$$\nabla_{u_m}\widehat{\mathrm{NMD}} = c \sum_{n\neq m} \frac{u_m - u_n}{\|u_m - u_n\|_2} = c\,(B-1)\,\hat{u}_m + O(1/t).$$
Substituting into the expression for $P_{u_m}$ in Theorem \ref{theorem:joint} and using $u_m = t\hat{u}_m$:
$$P_{u_m} = I - \frac{(t\hat{u}_m)\,(c(B-1)\hat{u}_m + O(1/t))^\top}{(2c/B)\,t + O(1)} = I - \frac{c(B-1)\,t\,\hat{u}_m \hat{u}_m^\top}{(2c/B)\,t} + O(1/t).$$
Canceling $t$ in the leading term:
$$P_{u_m} \;\longrightarrow\; I - \frac{B(B-1)}{2}\,\hat{u}_m \hat{u}_m^\top \quad \text{as } t \to \infty.$$ 
This is the orthogonal projector onto $\hat{u}_m^\perp$.
\end{proof}

\subsection{Bounded Action Spaces and Saturated Activations}
Lemma \ref{lemma:w2} assumes a Gaussian policy with linear mean. In practice, continuous-control policies typically apply a squashing nonlinearity such as $\tanh$ to enforce bounded actions, yielding \mbox{$a = \tanh(z)$ with $z \sim \mathcal{N}(\mu, \Sigma)$}. We discuss the impact on the preceding results.

\textbf{Effect on Lemma \ref{lemma:w2}.} The Wasserstein distance between two squashed Gaussians no longer admits a closed form: the push-forward of a Gaussian under $\tanh$ is non-Gaussian, and $W_2$ between such distributions must be computed numerically. However, in the unsaturated regime a standard Lipschitz inequality gives
$$W_2(\pi_m \circ \tanh^{-1}, \pi_n \circ \tanh^{-1}) \leq L_{\tanh}\, W_2(\pi_m, \pi_n) = L_{\tanh}\, \|\mu_m - \mu_n\|_2,$$
with $L_{\tanh} = 1$. \metric{} computed in pre-activation space therefore upper-bounds \metric{} computed in action space, and the relative ordering of behaviors is preserved.

\textbf{Effect on Theorem \ref{theorem:joint}.} The chain rule of Step 1 acquires an extra Jacobian factor $\mathrm{diag}(\mathrm{sech}^2(z_m))$. In the unsaturated regime this Jacobian is close to the identity and the projection structure of $P_{u_m}$ is preserved with multiplicative error $O(\|z_m\|^2)$. In the saturated regime the gradient flow $\nabla_{z_m} R$ along those coordinates is suppressed by the activation rather than the projector, so the projection $P_{u_m}$ contributes no further regulation in those directions. The corollary's qualitative regimes (vanishing/dominant deviation) carry over coordinate-wise, with the saturation acting as an additional, complementary gating mechanism rather than as a violation of the analysis.

\textbf{Effect on Corollary \ref{corollary:limits}.} Vanishing-deviation behavior (regime (i)) is unaffected, since \mbox{$\mu_m \to W_{\mathrm{shared}}\phi$} remains in the unsaturated regime for typical feature scales. Dominant-deviation behavior (regime (ii)) is strengthened: as $\|u_m\| \to \infty$, the squashing nonlinearity contributes additional suppression on top of the projector, ensuring that pathologically large LoRA outputs cannot escape regulation through the action nonlinearity.

\textbf{Scope of the guarantees.} For the implemented policy, the following holds. Proposition \ref{prop:nmd_distance} and Corollary \ref{corollary:limits} hold unconditionally, since they require only finite second moments. Lemma \ref{lemma:w2} and the exactness of the diversity constraint in Eq. \eqref{eq:alpha} hold in the pre-activation space, because the LoRA update enters only the mean of a shared-covariance Gaussian through the final linear layer. Theorem \ref{theorem:joint} and Corollary \ref{corollary:limits} hold exactly in pre-activation space and, after the tanh squashing, up to a multiplicative factor that is close to the identity in the unsaturated regime; in the saturated regime the activation suppresses the gradient along the saturated coordinates and the projector plays no role.

\section{Relationship of \metric{} and Agent-Indexed Diversity Metrics}\label{app:nmd_snd}

This appendix analyzes the relationship between \metric{} and agent-indexed diversity metrics, particularizing to SND~\citep{bettini2023system}. 


\textbf{\metric{} subsumes SND.} \metric{} and SND share the pairwise $W_2$ kernel; they differ in their domain. SND is a function on agent-indexed tuples $(\pi_1, \ldots, \pi_N) \in \mathcal{B}^N$, whereas \metric{} is a functional on probability measures over the behavior manifold, $\mathrm{\metric{}}(p) = \mathbb{E}_{\pi, \pi' \sim p}[d(\pi, \pi')]$ with $p \in \mathcal{P}(\mathcal{B})$ and $d$ as in Proposition~\ref{prop:nmd_distance}. Given a configuration $(\pi_1, \ldots, \pi_N)$, discard the agent labels and keep only which behaviors occur and how often, i.e., the empirical measure $\rho_N(\pi_1, \ldots, \pi_N) = \frac{1}{N}\sum_{i} \delta_{\pi_i}$. SND sums over unordered pairs and no term depends on which agent holds which policy, so it factors through this map: $\mathrm{SND} = \mathrm{\metric{}} \circ \rho_N$, up to a factor $\frac{N}{N-1}$ from self-pairs that the estimator in Eq.~\eqref{eq:nmd_approx} already accounts for. Every value SND reports is therefore a value \metric{} reports. The image of $\rho_N$, however, contains only uniform measures with $N$ atoms, whereas the domain of \metric{} also contains (a) atomic measures whose number of atoms differs from the number of agents, (b) non-uniform weights, which arise whenever the hypernetwork emits the same $(C_m, D_m)$ pair for several agents (the reason we index behaviors by $m$ rather than $i$), and (c) non-atomic densities, which the hypernetwork induces and Fig.~\ref{fig:main}d visualizes. For (c) there is no SND expression: SND evaluates a finite, agent-indexed sample, but it does not define the dispersion of the manifold itself, which is the quantity our framework learns, targets and controls.

\textbf{A reassignment example.} Consider two agents, two distinct behaviors $\pi_a, \pi_b$, and a single event at $\tau = T/2$. In system A, agent $1$ executes $\pi_a$ on $[0, \tau)$ and $\pi_b$ on $[\tau, T]$, while agent $2$ executes $\pi_b$ and then $\pi_a$. In system B, agent $1$ executes $\pi_a$ and agent $2$ executes $\pi_b$ throughout. At every $t$, both systems realize the same set of behaviors, so $\mathrm{\metric{}}(p_t) = d(\pi_a, \pi_b) > 0$ for both. An agent-indexed metric requires one policy per agent index over the episode, and the only such object that remains consistent under reassignment is the time-averaged policy $\bar{\pi}_i = \frac{1}{T}\int_0^T \pi_{i,t}\, dt$. For system A, $\bar{\pi}_1 = \bar{\pi}_2 = \frac{1}{2}\pi_a + \frac{1}{2}\pi_b$, hence $\mathrm{SND}(\bar{\pi}_1, \bar{\pi}_2) = 0$: a team that is diverse at every instant is reported as homogeneous, whereas system B is reported as diverse. \metric{} assigns the same value to both systems, which differ only in which agent executes which behavior, information on which the task does not depend. Evaluating SND instantaneously at each $t$ instead recovers $\mathrm{\metric{}} \circ \rho_N$, which is a special case of \metric{} and no longer a property of any individual agent.

\textbf{Well-posedness of diversity control.} SND is invariant to the number of equidistant agents~\citep{bettini2023system}, which allows a target to be set independently of team size in the symmetric case. Invariance of the metric value under a symmetric configuration, however, does not make a diversity constraint well-posed under an allocation that is time-varying, non-injective and of varying cardinality. In DiCo~\citep{bettini2024controlling}, the rescaled deviations are bound to agent indices, with one heterogeneous network per index; this is why DiCo's policies must be replicated to evaluate other team sizes in Fig.~\ref{fig:scaling_results}, and why an agent appearing mid-episode has no corresponding deviation. In contrast, $\mathrm{\metric{}}(p) = \mathrm{\metric{}}_{\mathrm{des}}$ is a single condition on $\mathcal{P}(\mathcal{B})$ that holds for every team size and every binding of behaviors to agents, and the deviations it constrains are generated on demand. This is what allows $\mathrm{\metric{}}_{\mathrm{des}}$ to transfer zero-shot across team sizes and target values (Fig.~\ref{fig:scaling_results}) and to be changed mid-episode (Fig.~\ref{fig:windflock_sequence}), with $p_t$ and $\mathrm{\metric{}}(p_t)$ remaining defined across agent removal without any re-indexing (Fig.~\ref{tab:combined_agent_removal}).

\textbf{Connecting \metric{} to roles, trajectories and task success.} \metric{} is defined on per-observation action distributions rather than on state visitation or task outcomes, so we verify in three ways that it controls functionally meaningful differences. First, an interventional test: Fig.~\ref{fig:scaling_results} (bottom row) varies $\mathrm{\metric{}}_{\mathrm{des}}$ at evaluation time and measures the completion rate. If \metric{} were decoupled from behavior, completion would be insensitive to it across the whole range; instead, each task exhibits a task-specific range of good target values (roughly $0.1$--$10$ for Reverse Transport and Navigation, $0.5$--$1.8$ for Dispersion, and $0.2$--$1.0$ for Football). Second, a state-space consequence of an action-space quantity: in Wind Flocking (Fig.~\ref{fig:windflock_sequence}), the inter-agent distance grows monotonically with the target between steps $20$ and $120$ and contracts again at step $170$ when the target is reduced, within one episode and without retraining. Third, roles: in Pressure Plate, the t-SNE projection of the behavior manifold (Fig.~\ref{fig:main}d) forms distinct clusters across $128$ replications that correspond to the phases of the task. Finally, Theorem~\ref{theorem:joint} implies that the projector fixes the magnitude of diversity but leaves the reward gradient free along every direction that preserves it, so that reward alone decides where the diversity budget is spent. This does not rule out wasted diversity, e.g., under redundant actuation, where different actions produce similar state transitions; extending \metric{} to state-visitation or trajectory-level distances is an interesting direction for future work.

\section{Experimental Details}
\label{app:expdetails}
We evaluate our approach across six distinct multi-agent tasks that challenge various aspects of coordination, observability, and physical interaction.

\textbf{Dispersion}: all agents spawn at the center of the arena while goals are distributed randomly. The complexity lies in the fact that with total observability and a shared starting point, agents must coordinate to ensure they do not all converge to the same closest goal. The optimal policy requires high behavioral diversity to disperse and tackle different goals simultaneously (Fig.~\ref{fig:dispersion}). 

\textbf{Navigation}: introduces limited observability, requiring agents to utilize LiDAR sensors to explore the area and locate specific goals. Agents must navigate efficiently under time constraints, relying on LiDAR to locate their specific targets in an environment where information is sparse during the initial exploration phase (Fig.~\ref{fig:navigation}).

\textbf{Reverse Transport}: agents spawn inside a physical package and must exert collective force to push it toward a goal. Because one agent is barely able to move the mass alone, the team must synchronize their velocities to move the package effectively, receiving rewards proportional to the reduction in distance to the goal (Fig.~\ref{fig:reverse_transport}). 

\textbf{Football}: features a team of three agents playing against a programmed heuristic. Beyond coordinating an attack to score a goal (sparse reward), the random spawn locations force agents to switch between defensive positioning and counter-attacking. The challenge is amplified by the need to manage ball control and shooting lanes while competing against a consistent heuristic opponent (Fig.~\ref{fig:football}).

\textbf{Pressure Plate}: a logic-based task where three agents must navigate a locked door by utilizing two pressure plates. This necessitates a three-phase strategy: (i) one agent presses the first plate to allow the others to cross, (ii) those agents must then find and hold a second plate on the opposite side so the first agent can also cross, and (iii) all agents move to the final goal (Fig.~\ref{fig:pressure_plate}).

\textbf{Wind Flocking}: explores energy-efficient collective movement among heterogeneous agents exposed to a north-to-south wind. The agents must discover a formation where the larger agent shields the smaller one from the wind to optimize their performance and minimize energy expenditure (Fig.~\ref{fig:windflocking}).

Events are fixed per task, but the behavioral distinction they induce is learned, and differs per task.
In Navigation, agents switch from exploration to target approach once the LiDAR detects the target. In Dispersion, the event induces commitment to distinct target subsets; without re-querying, agents converge on overlapping subsets. In Reverse Transport, the team configuration around the package adapts to the force-to-mass regime, which is why our method outperforms the baselines in episode length rather than in completion rate. In Pressure Plate, the three-phase role structure is described in Section \ref{subsec:adapt}. In Football, behaviors transition between defensive positioning and counterattack as a function of possession and ball proximity. In Wind Flocking, the inter-agent distance and the shielding geometry respond to the diversity target (Fig. \ref{fig:windflock_sequence}).

\begin{figure}[htbp]
    \centering
    
    \begin{subfigure}[b]{0.45\textwidth}
        \centering
        \includegraphics[width=\linewidth]{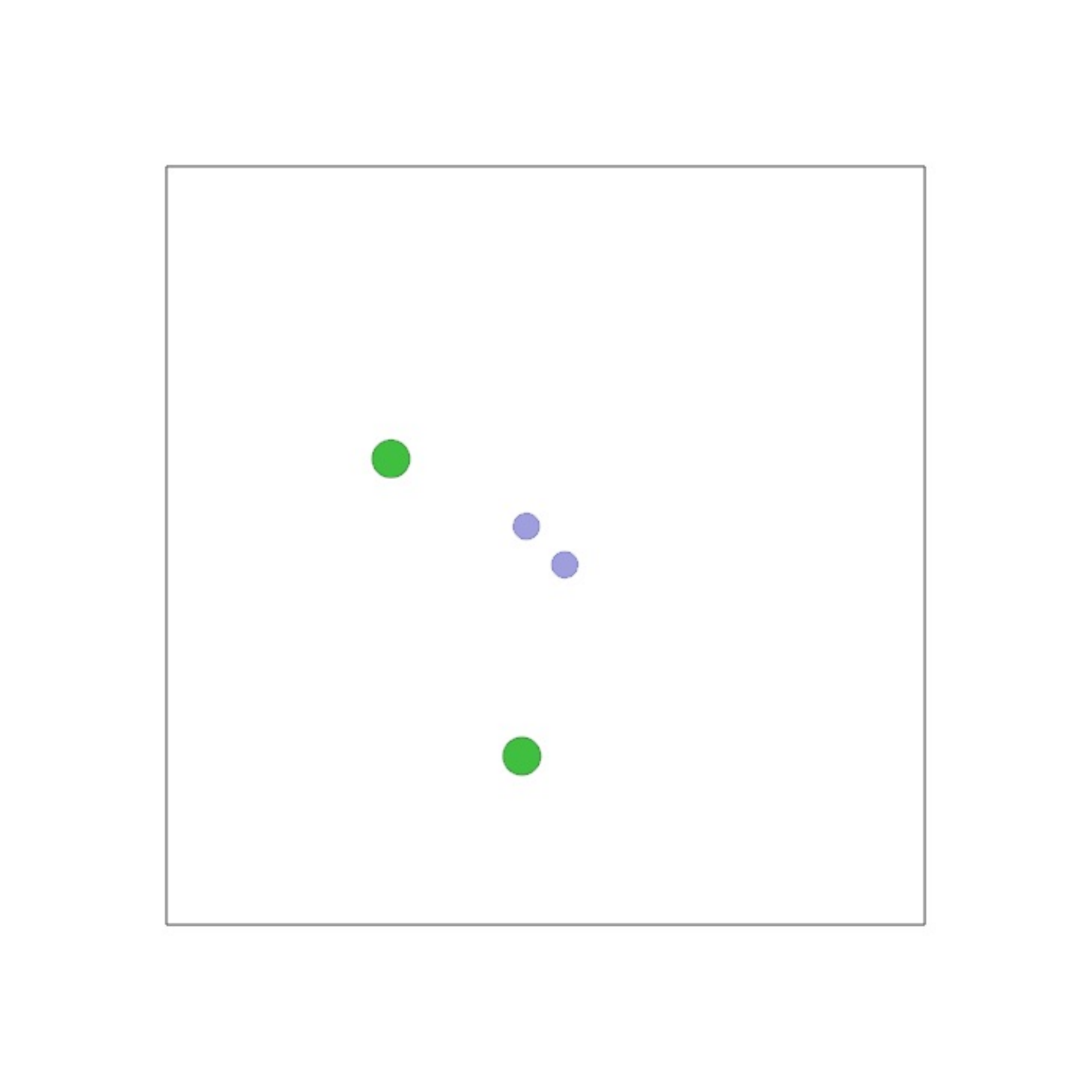} 
        \caption{\textbf{Dispersion}}
        \label{fig:dispersion}
    \end{subfigure}
    \hfill
    \begin{subfigure}[b]{0.45\textwidth}
        \centering
        \includegraphics[width=\linewidth]{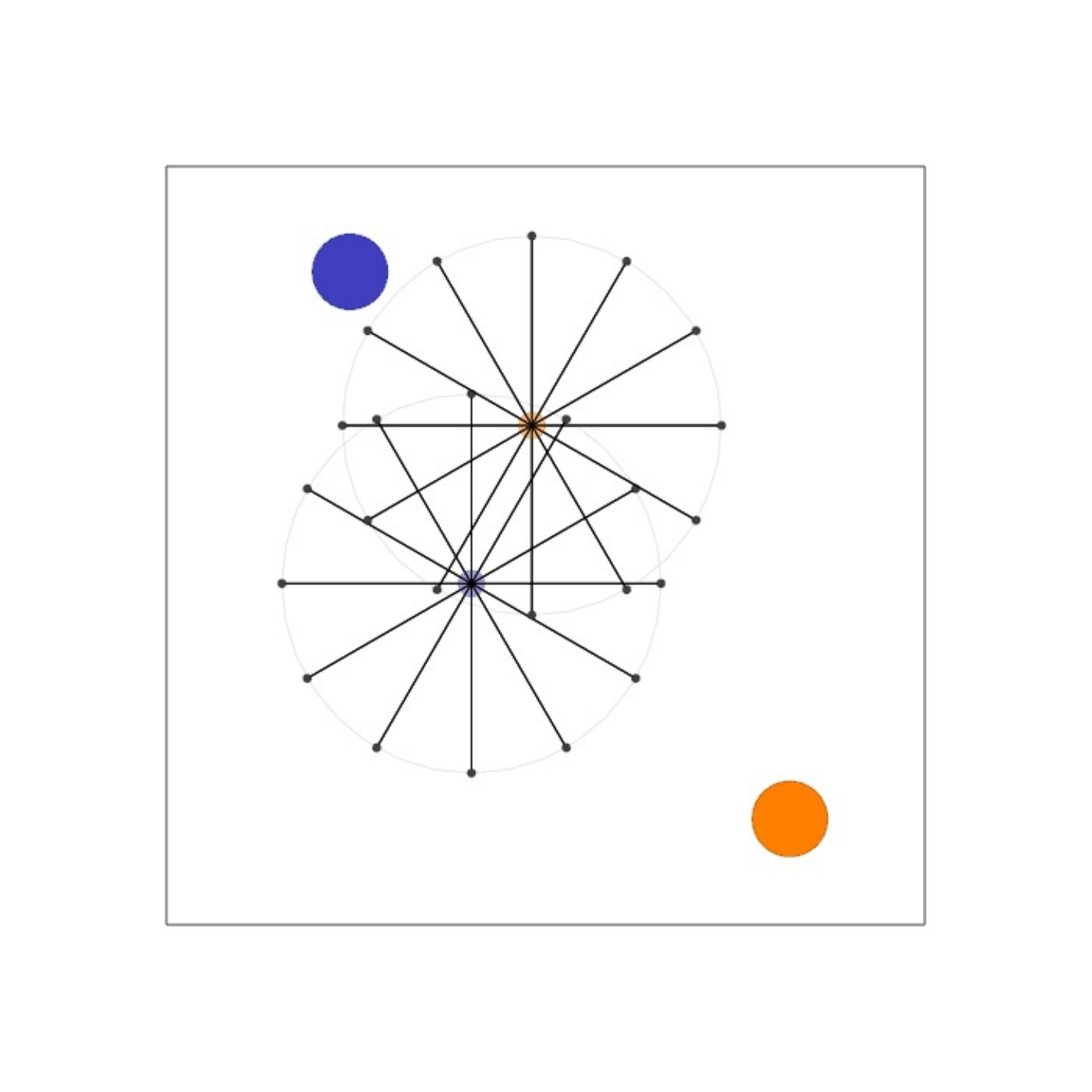}
        \caption{\textbf{Navigation}}
        \label{fig:navigation}
    \end{subfigure}
    
    \vspace{0.2cm} 
    
    \begin{subfigure}[b]{0.45\textwidth}
        \centering
        \includegraphics[width=\linewidth]{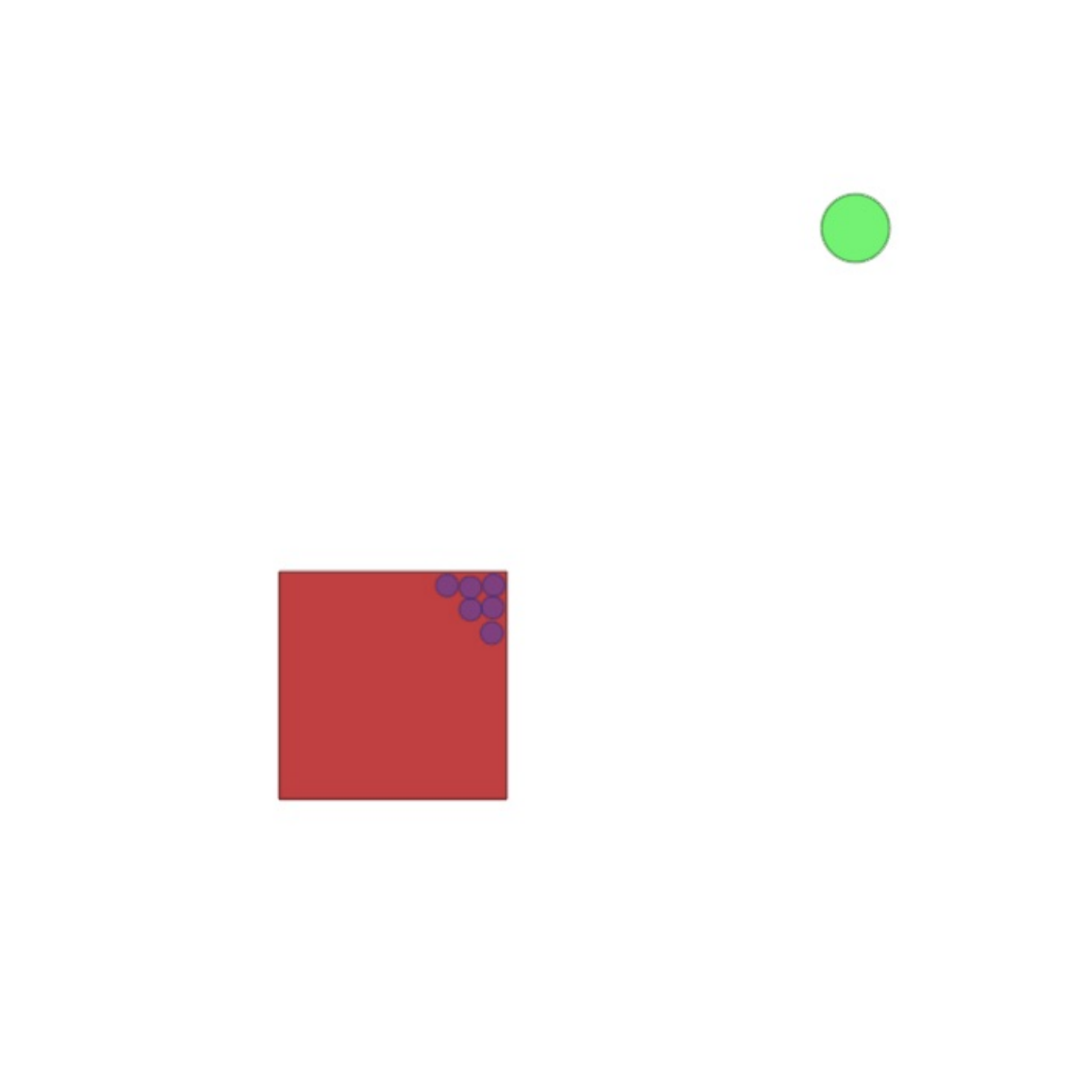}
        \caption{\textbf{Reverse Transport}}
        \label{fig:reverse_transport}
    \end{subfigure}
    \hfill
    \begin{subfigure}[b]{0.45\textwidth}
        \centering
        \includegraphics[width=\linewidth]{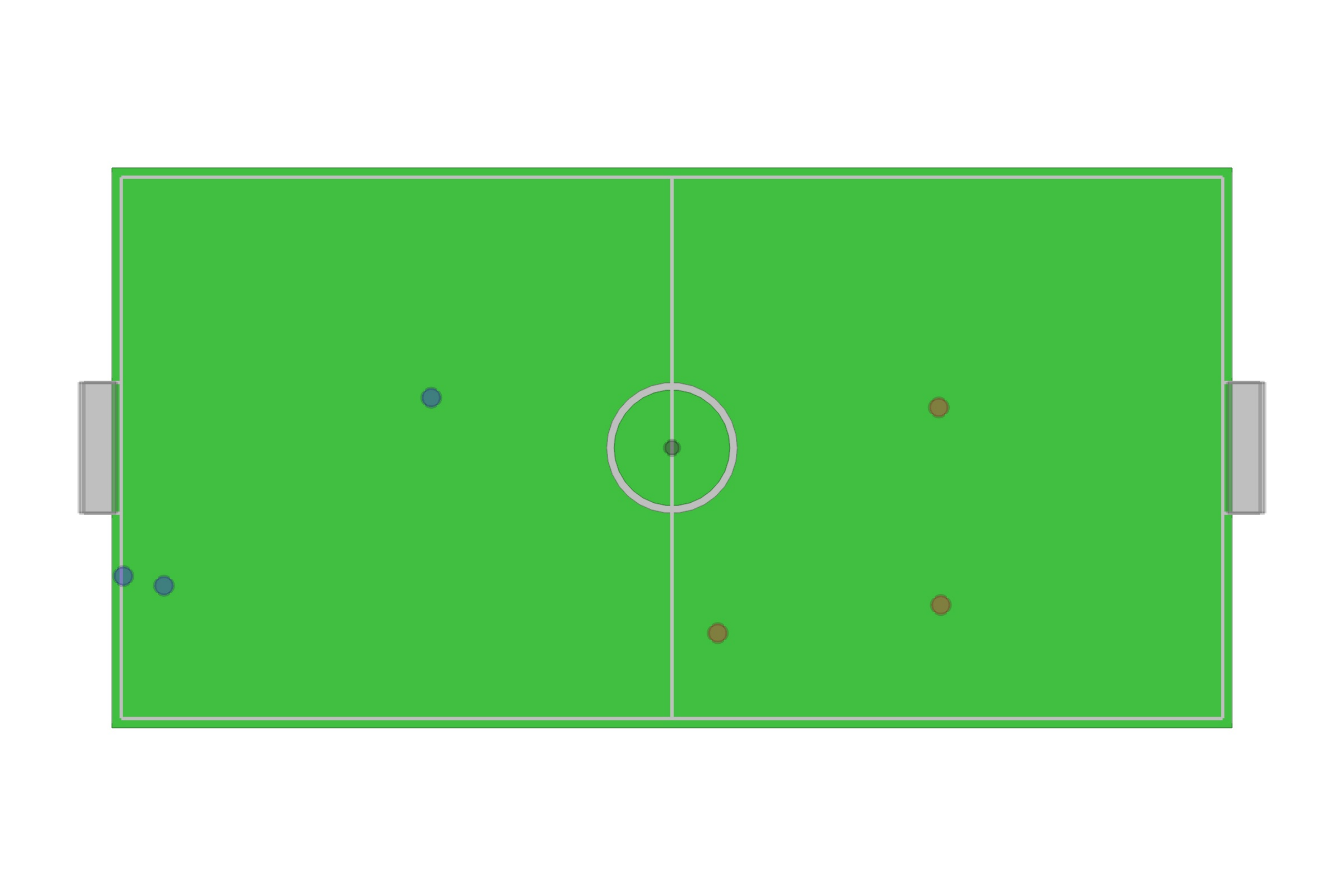}
        \caption{\textbf{Football}}
        \label{fig:football}
    \end{subfigure}
    \vspace{0.2cm} 
    
    \begin{subfigure}[b]{0.45\textwidth}
        \centering
        \includegraphics[width=\linewidth]{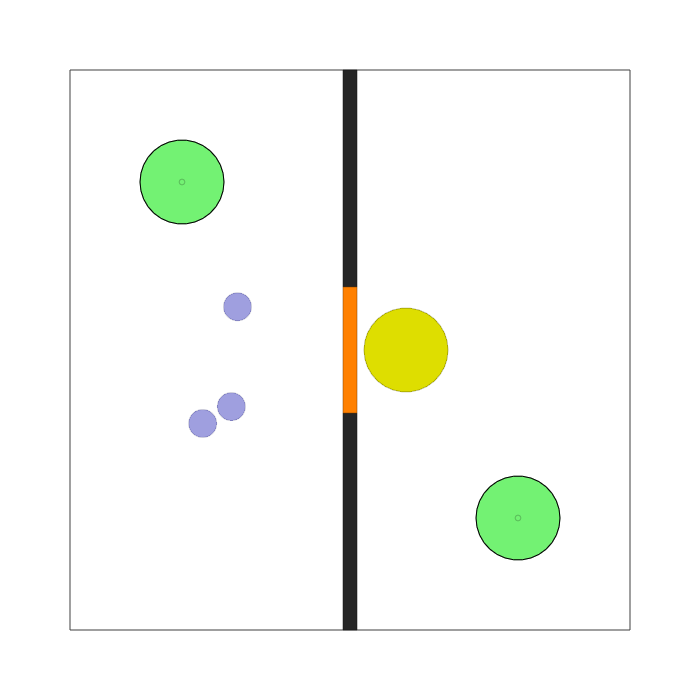}
        \caption{\textbf{Pressure Plate}}
        \label{fig:pressure_plate}
    \end{subfigure}
    \hfill
    \begin{subfigure}[b]{0.45\textwidth}
        \centering
        \includegraphics[width=\linewidth]{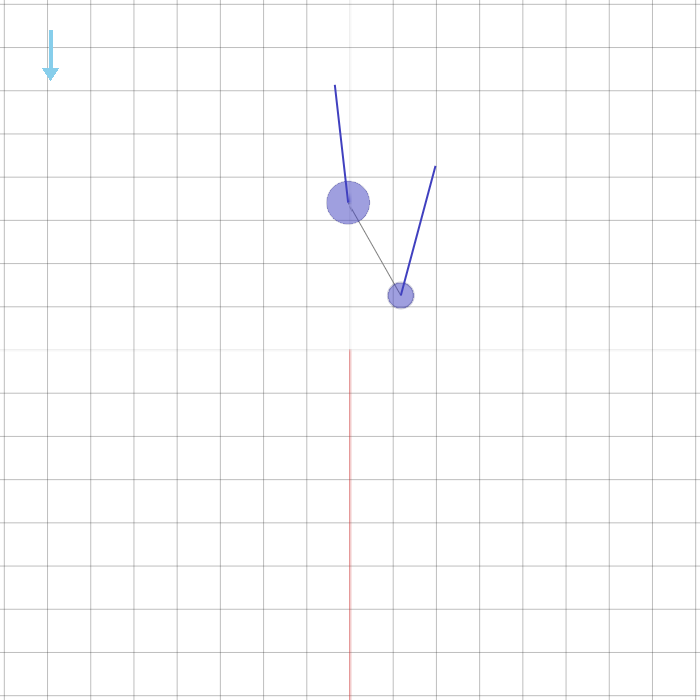}
        \caption{\textbf{Wind Flocking}}
        \label{fig:windflocking}
    \end{subfigure}
    
    \caption{Overview of the six tasks: (a) Dispersion, (b) Navigation, (c) Reverse Transport, (d) Football, (e) Pressure Plate, and (f) Wind Flocking.}
    \label{fig:six_tasks}
\end{figure}

\begin{table}[hpb]
\centering
\caption{Hyperparameters used in the experiments.}
\label{tab:all_hyperparameters}
\begin{tabular}{ll}
\toprule
\textbf{Hyperparameter} & \textbf{Value / Configuration} \\
\midrule
\multicolumn{2}{l}{\textbf{Shared Environment \& Execution Setup}} \\
\midrule
Training Episodes for Navigation \& Dispersion        & 400 \\
Training Episodes for Reverse Transport              & 300 \\
Training Episodes for Football                       & 2000 \\
Training Episodes for Pressure Plate \& Wind Flocking & 200 \\
Parallel Environments                                & 128 \\
Parallel Environments for Football                   & 256 (DiCo: 192) \\
Time steps for Navigation \& Dispersion              & 200 \\
Time steps for Pressure Plate \& Wind Flocking       & 200 \\
Time steps for Reverse Transport                     & 200 (CASH, HyperMARL: 300) \\
Time steps for Football                              & 300 \\
\midrule
\multicolumn{2}{l}{\textbf{Shared MAPPO \& Optimization Parameters}} \\
\midrule
Optimizer                        & Adam \\
Adam $\beta_1$                   & $0.9$ \\
Adam $\beta_2$                   & $0.999$ \\
Discount Factor ($\gamma$)       & $0.99$ (DiCo on Football: $0.9$) \\
GAE Parameter ($\lambda$)        & $0.95$ (DiCo on Football: $0.9$) \\
Clipping Parameter ($\epsilon$)  & $0.2$ \\
Max Gradient Norm (policy)       & $0.5$ (DiCo on Football: $5.0$) \\
\midrule
\multicolumn{2}{l}{\textbf{Proposed Framework}} \\
\midrule
Homogeneous Network Size         & [128, 128] (Wind Flocking: [64, 64]) \\
Transformer Layers (Blocks)      & $2$ \\
Attention Heads                  & $2$ \\
Embedding Dimension              & $64$ \\
MLP Hidden Dimension             & $256$ \\
LoRA Rank                        & $8$ \\
\midrule
\multicolumn{2}{l}{\textbf{DiCo Baseline \cite{bettini2024controlling}}} \\
\midrule
Homogeneous Network Size         & [128, 128] (Football: [256, 256]) \\
Heterogeneous Network Size (Per-Agent) & [128, 128] (Football: [256, 256]) \\
SND Values Tested (Grid Search) & $0.1, 0.2, 0.5, 0.9, 1.0, 1.1, 1.5, 2.0$ \\
Best SND for Navigation, Dispersion \& Pressure Plate & 1.0 \\
Best SND for Reverse Transport  & 0.1 \\
Best SND for Football  & 0.2 \\
\midrule
\multicolumn{2}{l}{\textbf{HyperMARL Baseline \cite{Tessera2024}}} \\
\midrule
Network Size (Hypernetwork)   & 64 \\
Network Size (Agents)         & [64, 64] \\
Embedding Dimension ($e_t^i$) & 8 (Reverse Transport: 16) \\
\midrule
\multicolumn{2}{l}{\textbf{CASH Baseline \cite{Fu2025}}} \\
\midrule
Policy Hidden Dimensions         & [128, 128] \\
Critic Hidden Dimension          & [128, 128] \\
GRU Hidden Dimension             & 128 (1 layer) \\
Learning Rate                    & $2 \times 10^{-3}$ \\
Hypernetwork Hidden Dimension    & $64$ \\
Hypernetwork Layers              & $4$ \\
Decoder Hidden Dimension         & $64$ \\
Use Two-Layer Decoder            & True \\
\bottomrule
\end{tabular}
\end{table}

For initial experimentation, we required around 500 compute hours using an NVIDIA GeForce RTX 2080 Ti GPU and an Intel(R) Xeon(R) Gold 6248R CPU @ 3.00GHz. For the final multi-seed evaluation, we invested additional 500 compute hours on the same hardware. The main hyperparameters used in the paper can be found in Table \ref{tab:all_hyperparameters}. For HyperMARL, we used the original implementation from the authors' repository with their hyperparameters where available~\citep{Tessera2024}.

The denominator of Eq. \eqref{eq:alpha} is an empirical estimator over the $B = 128$ behaviors in the parallel batch and therefore carries estimator variance. To reduce it, the implementation replaces the batch estimate with an exponential moving average of recent batch estimates, with discount factor of $0.9$ and initialized at $\metric{}_{\text{des}}$. The identity between realized and target diversity in Section \ref{subsec:dico} holds exactly with respect to this smoothed estimator; the smoothing introduces a tracking lag that is largest early in training, when diversity changes fastest, and negligible over runs of millions of steps. Immediately after an event fires, the re-query changes $u_m$ and $\alpha$ is recomputed from the same batch in the same forward pass, so the identity is restored at that time step rather than converged to.

\section{More Results}\label{app:more_results}

\subsection{Average Episode Rewards for Main Experiments}
\label{app:meanreward}
For all reported results in the main text, we evaluated the models using the best-performing checkpoints saved during the training process. Fig.~\ref{fig:mean_rewards_all_tasks} illustrates the learning progress across our six evaluation environments: Navigation, Dispersion, Reverse Transport, Football, Pressure Plate, and Wind Flocking.
The plots compare our method against the four baselines: 

\textbf{Full Parameter-Sharing.} A homogeneous baseline with no behavior allocation, where all agents share the same policy. 

\textbf{HyperMARL \citep{Tessera2024}.} Allocates behaviors episode-wise at timestep $0$ via a hypernetwork conditioned on agent identities.

\textbf{Diversity Control  (\textbf{DiCo})~\citep{bettini2024controlling}} Operates on fixed agent indices using cross-agent SND as the constraint; for each task, we perform a grid search over target SND values and report the best-performing variant. 

\textbf{Capability-Aware Shared Hypernetworks (\textbf{CASH})~\citep{Fu2025}. } Allocates behaviors at every timestep via a hypernetwork conditioned on capabilities when available, or agent identities

In nearly all environments, our method and the majority of baselines exhibit stable convergence toward high mean rewards, demonstrating that training is well-configured. The exception is CASH (Fig.~\ref{fig:nav_mean} and~\ref{fig:dispersion_mean}), which exhibits ``shaky'' performance and, in some independent runs, fails to converge to a competitive reward  (Fig.~\ref{fig:reverse_transport_mean}).
We attribute this behavior to the policy brittleness inherent to its architecture. Because CASH requires re-querying the hypernetwork at every timestep to generate agent-specific weights, the resulting policy can be highly sensitive to minor fluctuations in the hypernetwork's output.
For DiCo, we evaluated several SND values and picked the one with best performance to ensure a fair comparison; the specific values used for each environment are detailed in Table~\ref{tab:all_hyperparameters}.

\begin{figure*}[htbp]
    \centering
    \begin{subfigure}[b]{0.48\textwidth}
        \centering
        \includegraphics[width=\textwidth]{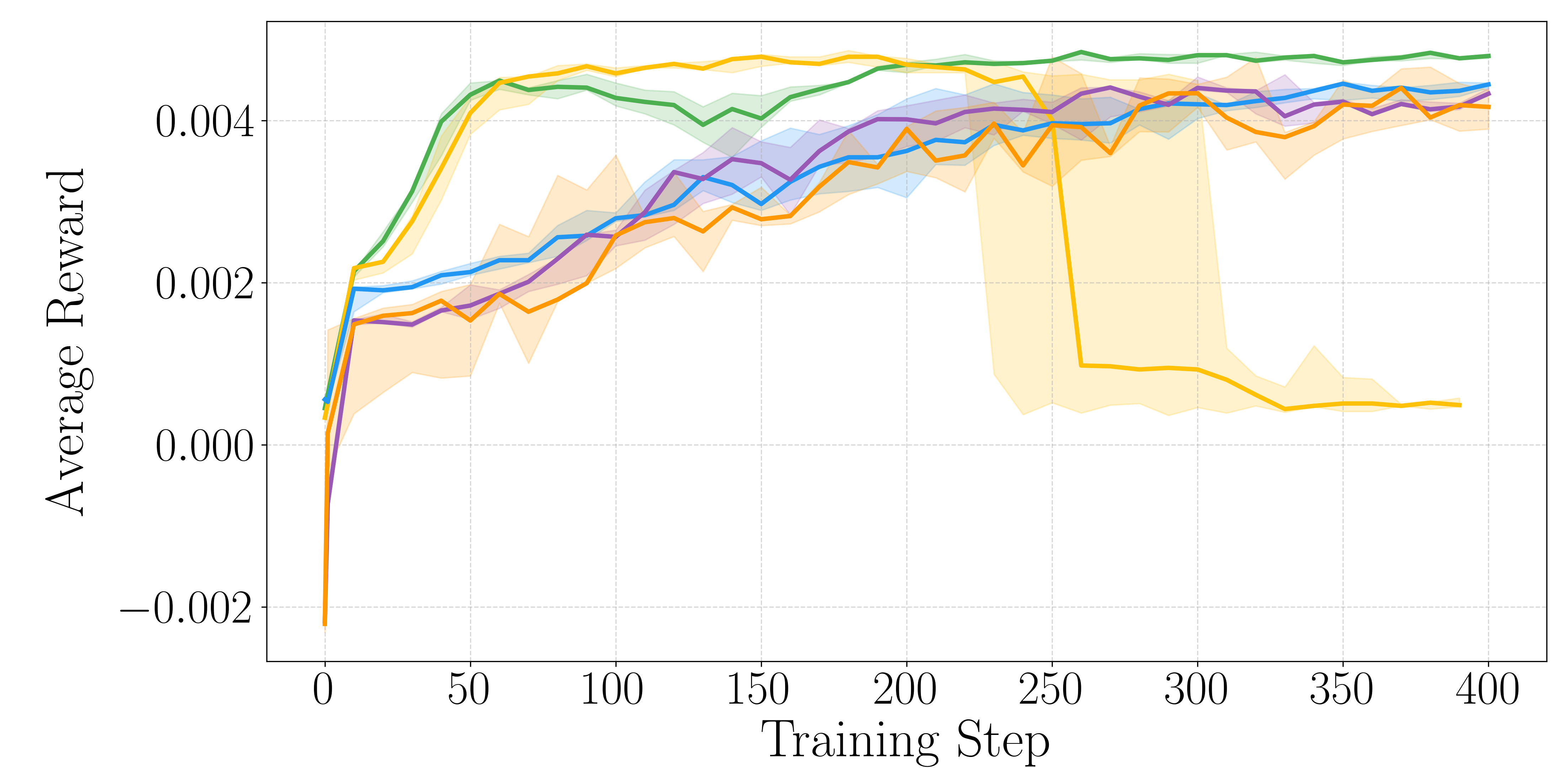}
        \caption{Navigation}
        \label{fig:nav_mean}
    \end{subfigure}
    \hfill
    \begin{subfigure}[b]{0.48\textwidth}
        \centering
        \includegraphics[width=\textwidth]{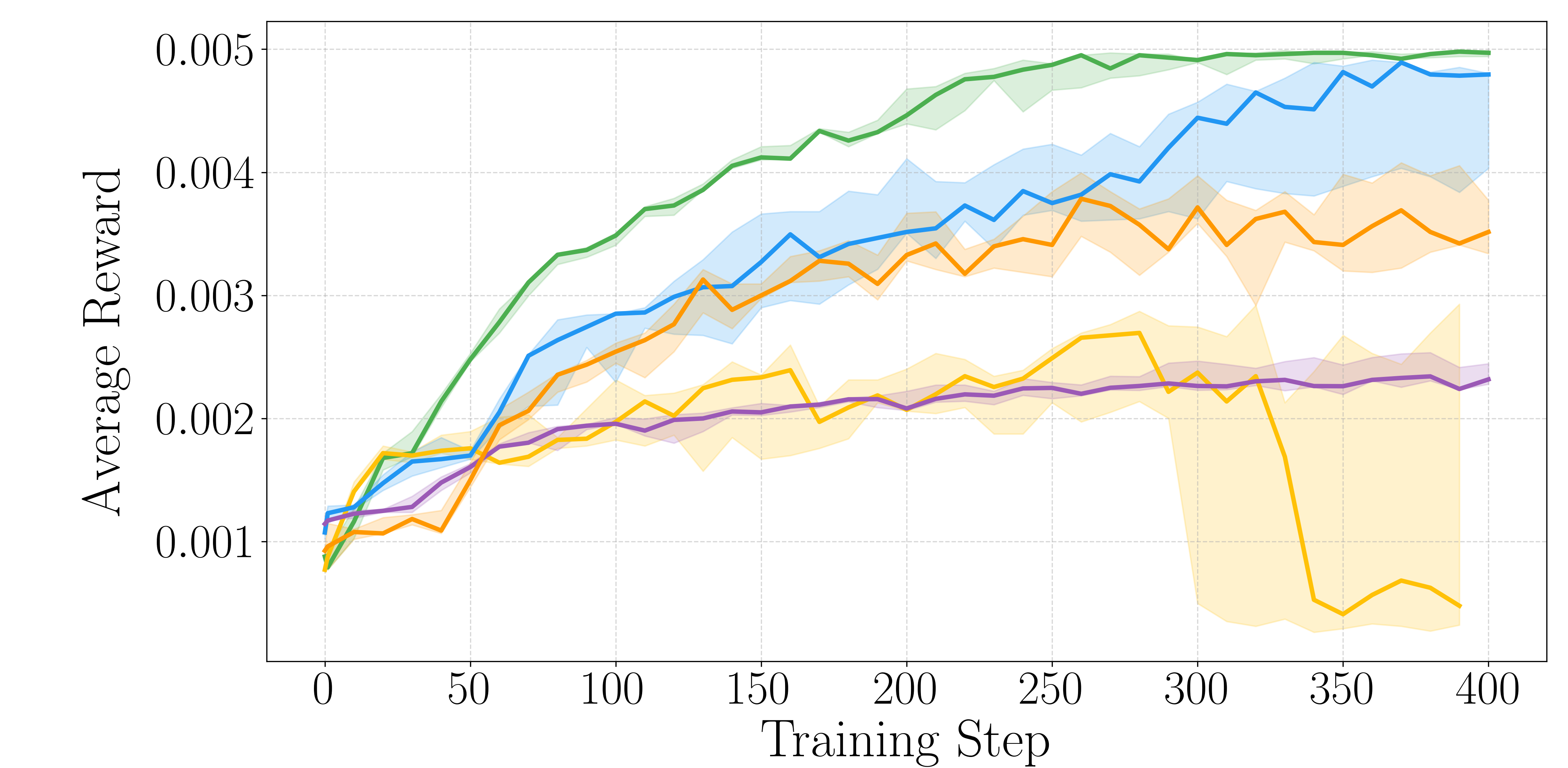}
        \caption{Dispersion}
        \label{fig:dispersion_mean}
    \end{subfigure}
    
    \vspace{0.4cm} 
    
    \begin{subfigure}[b]{0.48\textwidth}
        \centering
        \includegraphics[width=\textwidth]{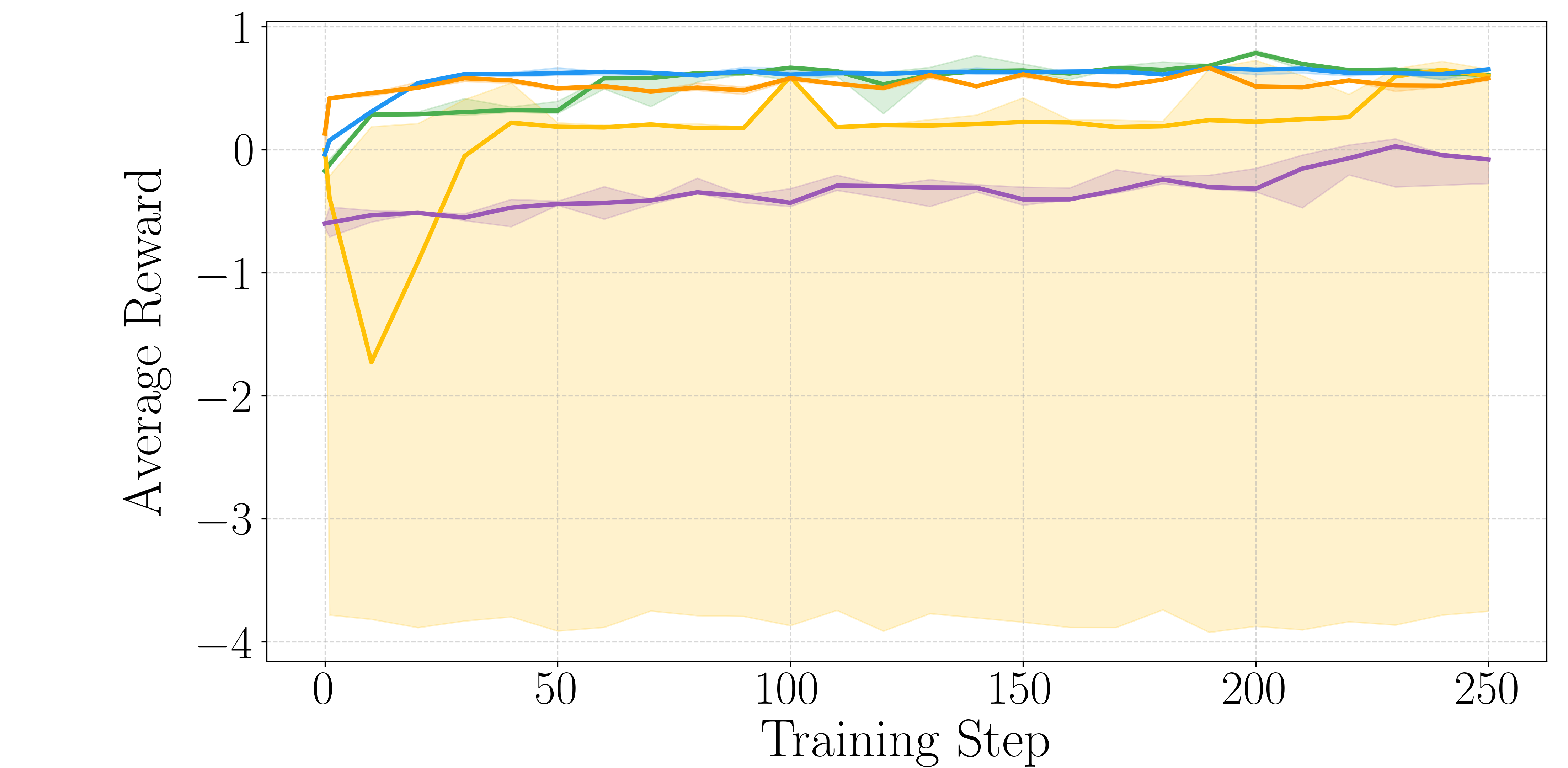}
        \caption{Reverse Transport}
        \label{fig:reverse_transport_mean}
    \end{subfigure}
    \hfill
    \begin{subfigure}[b]{0.48\textwidth}
        \centering
        \includegraphics[width=\textwidth]{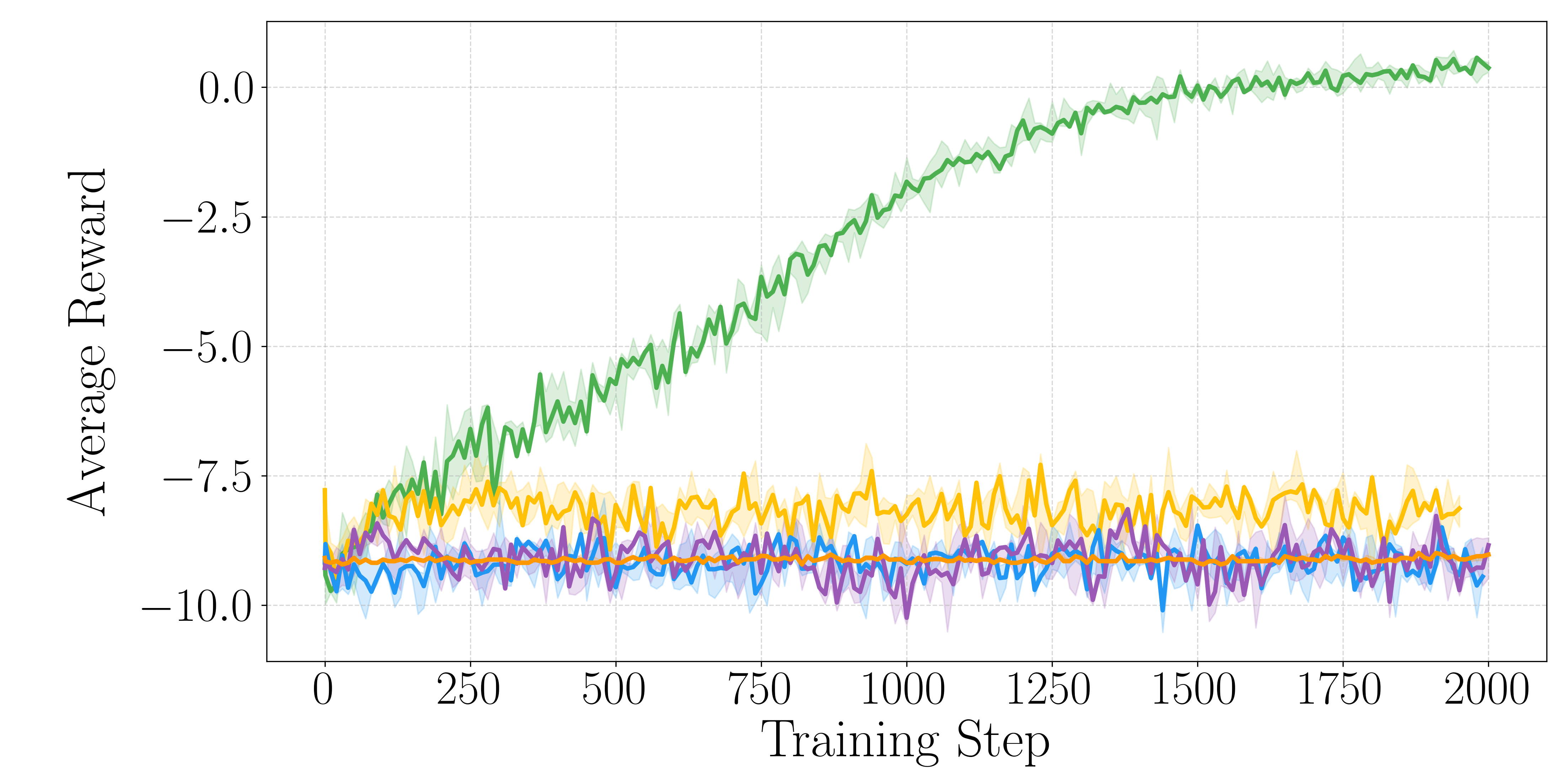} 
        \caption{Football}
        \label{fig:football_mean}
    \end{subfigure}
    
    \vspace{0.4cm} 
    
    \begin{subfigure}[b]{0.48\textwidth}
        \centering
        \includegraphics[width=\textwidth]{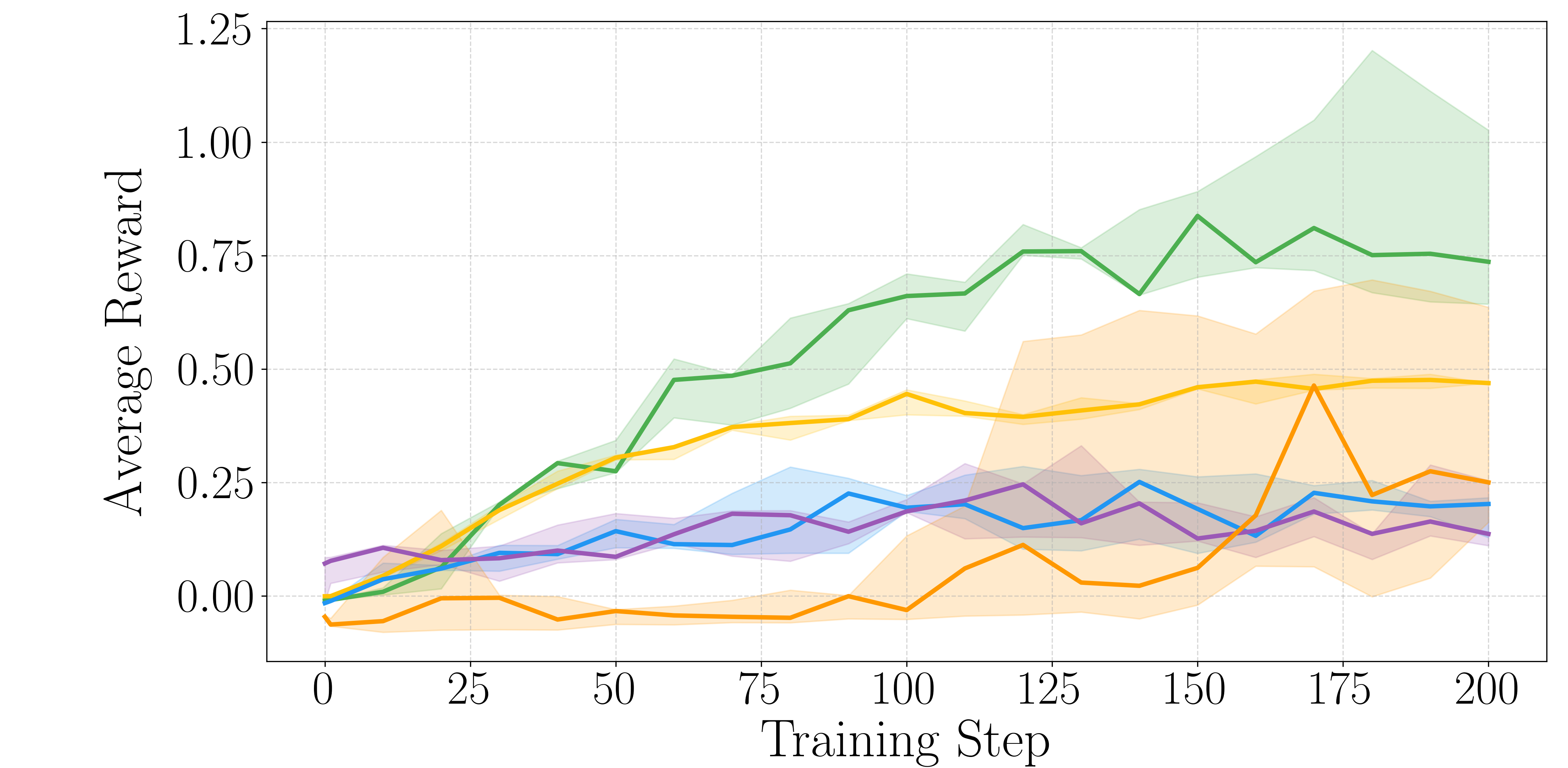} 
        \caption{Pressure Plate}
        \label{fig:pressure_plate_mean}
    \end{subfigure}
    \hfill
    \begin{subfigure}[b]{0.48\textwidth}
        \centering
        \includegraphics[width=\textwidth]{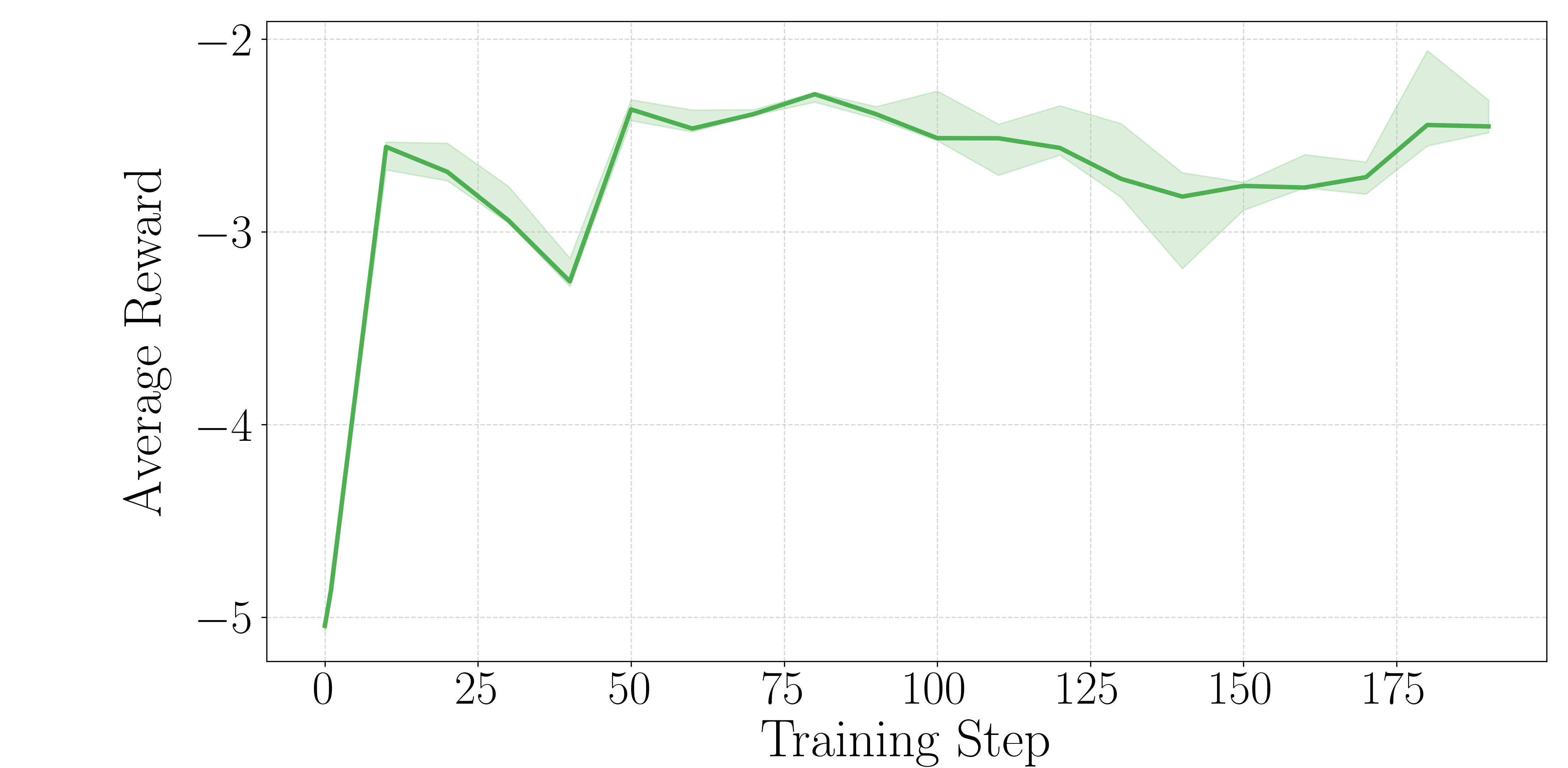} 
        \caption{Wind Flocking}
        \label{fig:wind_flocking_mean}
    \end{subfigure}
    
    \caption{\textbf{Comparison of mean rewards across all evaluation tasks.} Learning curves for all the tasks. Colors denote the different methods evaluated: Ours (green), Full Parameter-Sharing (purple), CASH ~\citep{Fu2025} (yellow),  HyperMARL ~\citep{Tessera2024} (orange), and DiCo ~\citep{bettini2024controlling} (blue). Performance is averaged over 5 seeds.}
    \label{fig:mean_rewards_all_tasks}
\end{figure*}

\subsection{Component Ablation}
\label{app:ablation}
Table~\ref{tab:ablation} reports completion rates for five ablated variants of our framework on Navigation, Dispersion and Reverse Transport: (A) removes the event mechanism by querying the hypernetwork only once at $t=0$; (B) replaces \metric{} with SND in the diversity constraint; (C) removes diversity control by fixing $\alpha = 1$, yielding an event-conditioned hypernetwork without \metric{}; (D) replaces the predicted LoRA factors with the parameters of a full MLP; and (E) applies LoRA to all layers of the heterogeneous part of the policy. Each variant holds the remaining components fixed.

\begin{table}[h]
\centering
\caption{\textbf{Component ablation.} Completion rate [\%] (median $\pm$ standard deviation over $5$ seeds). Each variant removes or replaces one component while holding the others fixed; configuration and hyperparameters match Fig.~\ref{fig:four_tasks_row} and Table\ref{tab:all_hyperparameters}.}
\label{tab:ablation}
\small
\setlength{\tabcolsep}{4pt}
\begin{tabular}{lccc}
\toprule
\textbf{Variant} & \textbf{Navigation} & \textbf{Dispersion} & \textbf{Reverse Transport} \\
\midrule
Ours (full) & $\mathbf{96.02 \pm 0.77}$ & $\mathbf{98.04 \pm 0.87}$ & $\mathbf{100.00 \pm 0.04}$ \\
(A) No event re-query (single query at $t=0$) & $87.58 \pm 3.18$ & $85.55 \pm 4.62$ & $100.00 \pm 23.34$ \\
(B) \metric{} replaced by SND & $84.90 \pm 14.82$ & $89.57 \pm 5.65$ & $99.61 \pm 30.48$ \\
(C) Event-conditioned, no diversity control ($\alpha = 1$) & $84.22 \pm 8.93$ & $90.64 \pm 2.82$ & $91.56 \pm 37.01$ \\
(D) Full MLP instead of LoRA & $86.74 \pm 3.88$ & $88.67 \pm 4.28$ & $61.41 \pm 34.15$ \\
(E) LoRA applied to all policy layers & $87.59 \pm 3.15$ & $92.52 \pm 4.32$ & $100.00 \pm 2.24$ \\
\bottomrule
\end{tabular}
\end{table}

Across all three tasks, the full method attains the highest (on Reverse Transport, tied-highest) median and, by a substantial margin, the lowest seed-to-seed variability ($\sigma \leq 0.87$, against $2.2$--$37.0$ for the variants). The medians alone understate this difference. On Reverse Transport, five of six configurations exceed a median of $91\%$, yet their standard deviations span $0.04$ to $37.0$: the variants do not degrade gracefully, but fail on a subset of seeds while solving the task on the rest. The full method succeeds on $5/5$ seeds, against $4/5$ for (A), $3/5$ for (B) and (C), and $2/5$ for (D).

The variants also differ in which task exposes them. (A) costs most on the two tasks whose relevant signal changes after $t=0$ ($-8.4$ points on Navigation, $-12.5$ on Dispersion), while leaving the median on Reverse Transport intact, consistent with re-querying mattering only once events occur. (B) produces the largest variance on Navigation ($\sigma = 14.8$, driven by one seed at $51.9\%$), so SND admits a distinct failure mode under dynamic allocation rather than uniformly weaker performance. (C), the event-conditioned model without diversity control, loses between $7.4$ and $11.8$ points of median completion on every task, showing that event conditioning alone does not account for the gains. (D) is the most damaging overall, reducing the median completion on Reverse Transport to $61.4\%$. (E) is the least critical, costing about $8$ points on Navigation; deeper adapters mainly harm training stability and forfeit the linearity that makes the diversity constraint exact (Section~\ref{subsec:dico}). Overall, each removal degrades a different task along a different axis: the components are complementary rather than redundant, and their principal joint effect is to make training reliable rather than to raise peak performance.

\subsection{LoRA Rank Ablation}
\label{app:lorarank}
We conducted an ablation study on the navigation task across five independent seeds to determine the optimal rank for the Low-Rank Adaptation (LoRA) modules. Based on these results, illustrated in Fig.~\ref{fig:lorarank}, we utilized a rank of $r=8$ for all subsequent experiments. This configuration proved to be the ``sweet spot'' for our architecture, providing sufficient expressivity while maintaining the stability of the hypernetwork during training.
In this sense, we recall that the dimensionality of the LoRA rank plays a critical role in the underlying optimization landscape. Restricting the hypernetwork to emit matrices $C_m \in \mathbb{R}^{r \times d}$ and $D_m \in \mathbb{R}^{d_a \times r}$, where $r \ll d$, significantly stabilizes optimization relative to generating full dense matrices.
\begin{figure}
    \centering
    \includegraphics[width=0.8\linewidth]{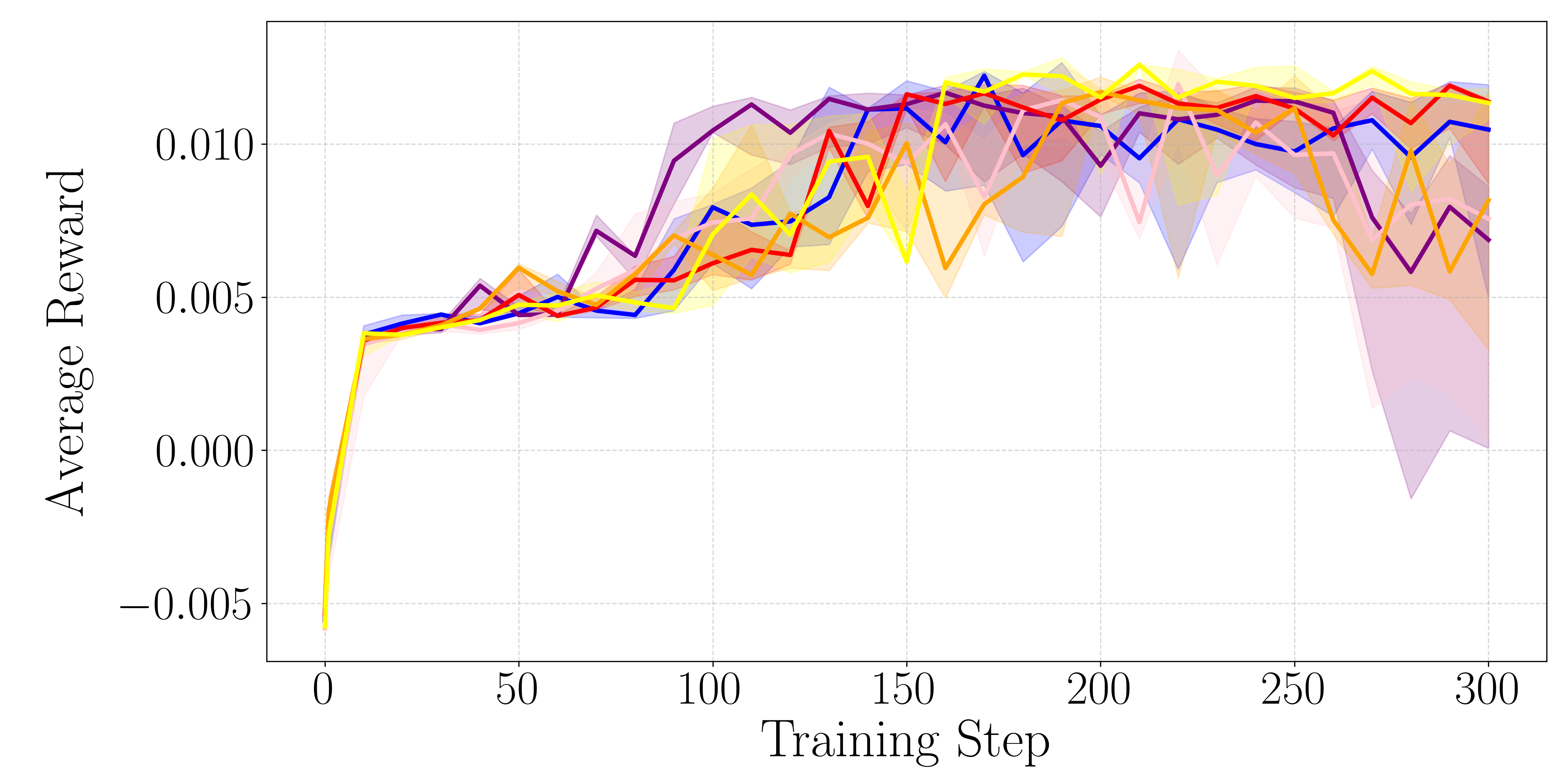}
    \caption{LoRA rank ablation in the navigation task. The plot illustrates the mean episode reward across ranks, evaluated over 5 seeds. Colors: Rank 2 (yellow), Rank 4 (orange), Rank 8 (red), Rank 16 (pink), Rank 32 (purple), and Rank 64 (blue).}
    \label{fig:lorarank}
\end{figure}

\section{Broader Societal Impact}\label{app:societal_impact}

The most direct beneficiaries of event-driven behavioral adaptation are safety-critical robotic teams operating in dynamic, partially observable environments---search and rescue, post-disaster reconnaissance, environmental monitoring, and inspection of hazardous infrastructure---where agents may fail mid-mission and the optimal allocation of behaviors changes abruptly as new information arrives. The agent-removal and unseen-event experiments in Section~\ref{sec:experiments} show that the team can absorb such disruptions without retraining, which is a precondition for deploying these systems outside the lab. Zero-shot generalization across team size, capability, and target diversity additionally removes the need to train a separate policy for each deployment configuration, lowering the compute and energy cost of multi-robot deployment, and exposing $\mathrm{NMD}_{\mathrm{des}}$ as an interpretable input scalar gives operators a post-hoc handle on team behavior without retraining. Beyond its direct technical contribution, treating behaviors as properties of the task rather than of agents carries a conceptual implication worth surfacing. Most prior MARL frameworks tie roles to agent identity, implicitly assuming that some agents are intrinsically suited to certain functions; in our framework, no agent is permanently the leader, the scout, or the helper, and any agent can instantiate any behavior the task requires at the moment it requires it. We see this as a more democratic model of cooperation, in which roles emerge from circumstance rather than from identity. 

On the negative side, improvements in cooperative MARL are dual-use: the same properties that benefit search-and-rescue can in principle be applied to coordinated autonomous swarms in military or surveillance contexts. We do not believe our contribution unlocks fundamentally new capabilities along this axis but we acknowledge that more reliable team coordination is a building block that can be repurposed. A separate concern, noted as a limitation in Section~\ref{sec:conc}, is that our hypernetwork is centralized and ingests observations from all agents, which raises both privacy considerations (when observations include sensor data from environments containing people) and security considerations (a single point of failure for the team's behavioral allocation); the decentralized extension discussed in Section~\ref{sec:conc} is therefore a meaningful mitigation as well as a methodological direction. 

\newpage

\end{document}